%% file: main.tex
\documentclass{article}

\usepackage{amssymb}
\usepackage{amsmath}
\usepackage{amsthm}
\usepackage{url}
\usepackage{graphics}
\usepackage{wrapfig}
\usepackage{graphicx}
\usepackage{pdfpages}
\usepackage{color}
\usepackage{enumerate}
\usepackage{stmaryrd}
\usepackage{comment}
\usepackage{textcomp}
\usepackage{listings}
\usepackage{tikz}
\usetikzlibrary{shapes}
\usepackage{tikz-qtree}
\usepackage{xcolor}
\usepackage{enumitem}
\usepackage{lmodern}

\usetikzlibrary{matrix}

\include{macros}

\include{pl}

\theoremstyle{plain}
\newtheorem{theorem}{Theorem}[section]
\newtheorem{lemma}[theorem]{Lemma}
\newtheorem{proposition}[theorem]{Proposition}

\theoremstyle{definition}
\newtheorem{definition}[theorem]{Definition}

\theoremstyle{remark}

\newtheorem{example}[theorem]{Example}

\author{}
\title{Expressiveness within Sequence Datalog}
\date{}

\begin{document}

\author{Heba Aamer \\ 
\normalsize Hasselt University \\ 
\normalsize heba.mohamed@uhasselt.be \and 
Jan Hidders \\ 
\normalsize Birkbeck, University of London \\ 
\normalsize jan@dcs.bbk.ac.uk \and 
Jan Paredaens \\ 
\normalsize Universiteit Antwerpen \\ 
\normalsize jan.paredaens@uantwerpen.be \and
Jan Van den Bussche \\ 
\normalsize Hasselt University \\ 
\normalsize jan.vandenbussche@uhasselt.be}

\maketitle

\begin{abstract}

Motivated by old and new applications, we investigate Datalog as
a language for sequence databases.  We reconsider classical
features of Datalog programs, such as negation, recursion,
intermediate predicates, and relations of higher arities.  We
also consider new features that are useful for sequences,
notably, equations between path expressions, and ``packing''.
Our goal is to clarify the relative expressiveness of all these
different features, in the context of sequences.
Towards our goal, we establish a number of redundancy
and primitivity results, showing that certain features can, or
cannot, be expressed in terms of other features.  These results
paint a complete picture of the expressiveness relationships
among all possible Sequence Datalog fragments that can be formed
using the six features that we consider.

This paper is the extended version of a paper presented at PODS
2021~\cite{seqdatalog-pods}
\end{abstract}



\section{Introduction}

Interest in sequence databases dates back for at least three
decades \cite{3genmanifesto}.  For clarity, here, by sequence
databases, we do not mean relations where the tuples are ordered
by some sequence number or timestamp, possibly arriving in a
streaming fashion (e.g.,
\cite{chomicki_tql,ramaseq,ramasrql,zaniolo_streams_tods}).
Rather, we mean databases that allow the management of large
\emph{collections of sequences}.

In the early years, sequence databases were motivated by
applications in object-oriented software engineering
\cite{oomanifesto} and in genomics
\cite{bonnermecca_sequences,jo_dmls}.  While these applications
remain relevant, more recent applications
of sequence databases include the following.
\begin{itemize}
\item Process mining \cite{pmanifesto} operates on event logs,
which are sets of sequences.  Thus, sequence databases, and
sequence database query languages, can serve as enabling technology for
process mining.  For example, a typical query one may want to be
able to support is look for all logs in which every occurrence of
`complete order' is followed by `receive payment'.
\item
Graph databases have as main advantage over relational databases
that they offer convenient query primitives for retrieving paths.
Paths are, of course, sequences.  For example, the G-CORE graph
query language proposal \cite{g-core} supports the querying of
sequences stored in the database, separately from the graph;
these sequences do not even have to correspond to actual paths in the
graph.  An example query in such a context could be to return the
nodes that belong to all paths in a given set of paths.
\item
JSON Schema \cite{chili_jsonschema}
is based on the notion of JSON pointers, which are
sequences of keys navigating into nested JSON objects.
In J-Logic \cite{hidders2017j} we showed that modeling JSON
databases as sequence databases is very convenient for
defining JSON-to-JSON transformations in a logical, declarative
manner.

For a simple example, consider a JSON object $\it Sales$ that is
a set of key--value pairs, where keys are items; the value for an
item is a nested object holding the sales volumes for the item by
year.  Specifically, the nested object is again a set of
key--value pairs, where keys are years and values are numbers.
We can naturally view $\it Sales$ as a set of length-3 sequences
of the form item--year--value.  Restructuring the object to group
sales by year, rather than by item, then simply amounts to
swapping the first two elements of every sequence.  For another
example, checking if two JSON objects are deep-equal amounts to
checking equality of the corresponding sets of sequences.
\item
Logical approaches to information extraction
\cite{xlog,documentspanners}
model the result of
an information extraction as a sequence database.
\end{itemize}

Given the importance of sequences in various advanced database
applications, our research goal in this paper is to obtain a
thorough understanding of the role that different language features
play in querying sequence databases.  For such an investigation,
we need an encompassing query language in which these features
are already present, or can be added.  For this purpose we adopt
Datalog, a logical framework that is well established in database
theory research, and that has continued practical relevance
\cite{datalogreloaded,datalog2.0,datalog2.019}.

Indeed, Datalog for sequence databases, or Sequence Datalog, was
already introduced and studied by Bonner and Mecca in the late
1990s \cite{bonnermecca_sequences,meccabonner_termination}.  They
showed that, to make Datalog work with sequence databases, all we
have to do is to add terms built from sequence variables using
the concatenation operator.  In our work we refer to such terms
as \emph{path expressions} and refer to sequence variables as
\emph{path variables}.\footnote{We actually work with a minor
variant of Bonner and Mecca's language; while they additionally
introduce index terms, but only allow path expressions in the
heads of rules, we allow path expressions also in rule bodies,
and additionally introduce atomic variables. The two variants are
equivalent in that one can be simulated by the other requiring no
additional features such as negation or recursion.}
Bonner and Mecca studied computational completeness, complexity,
and termination guarantees for Sequence Datalog, and showed how
to combine Sequence Datalog with subcomputations expressed using
transducers.

Sequence Datalog was recently also considered for information
extraction (``document spanners''), with regular expression
matching built-in as a primitive \cite{spannerlog,peterprog}.
Such regular expressions may be viewed as very useful syntactic
sugar, as they are also expressible using recursion.  Adding regular
matching directly may be compared to Bonner and Mecca's
transducer extensions; the PTIME capturing result reported by
Peterfreund et al.\ \cite{peterprog} may be compared to
Corollary~3 of Bonner and Mecca \cite{bonnermecca_sequences}.

In the present work, we study the relative expressiveness of
query language features in the context of Sequence Datalog.  Some
of the features we consider are standard Datalog, namely,
recursion, stratified negation, and intermediate predicates.  The
latter feature actually comprises two features, since we
distinguish between monadic intermediate predicates and
intermediate predicates of higher arities.  While we omit regular
expression matching as a feature, we consider two further
features that are specific to sequences:

\begin{itemize}
\item
Equalities between path expressions, which we call
\emph{equations}, allow for the elegant expression of
pattern matching on sequences.
\item
\emph{Packing}, a feature introduced in J-Logic, is a versatile
tool that allows for subsequences to be ``bracketed'' and
temporarily treated as atomic values; they can be unpacked later.
\end{itemize}

The standard Datalog features, whose expressiveness is well
understood on classical relational structures
\cite{ahv_book,ef_fmt2}, need to be re-examined in the
presence of sequences; moreover, their interaction with the new
features needs to be understood as well.  For example, consider
recursion versus equations, and the query that checks whether an
input sequence $\$x$ consists exclusively of $a$'s.  (Path
variables are prefixed by a dollar sign.)  With an equation we
can simply write $\$x \cdot a = a \cdot \$x$ (using the dot for
concatenation).  Without equations (or other means to simulate
equations), however, this query can only
be expressed using recursion.  For another example, consider
monadic versus higher-arity intermediate predicates.
Classically, there are well-known arity hierarchies for Datalog
\cite{grohe_arity}.  In our setting, however, a unary relation
can already hold arbitrary-length sequences, and indeed, using a
simple coding trick, we will see that the arity feature is
actually redundant.

In our work, we have chosen to define expressiveness in terms of
the baseline class of ``flat unary queries'', namely, functions
from unary relations to unary relations, where both the input and
the output are just sets of plain, unpacked sequences.  In this
way, we avoid trivial tautologies such as ``packing is a
primitive feature, because without it, we cannot create packed
sequences''.  As a matter of fact, we will show that packing, although it
certainly is a convenient feature, is actually redundant for
expressing these flat unary queries.  A result in this direction
was already stated for J-Logic \cite{hidders2017j}, but the
technique used there to simulate packing requires recursion.  In the
present paper, we show that packing is redundant also in the
absence of recursion.  Our proof technique leverages associative
unification \cite{abdulrab1989solving}, and more specifically,
the termination of associative unification for particular cases
of word equations \cite{duran2018associative}.

Our further results can be summarized as follows.

\begin{enumerate}
\item
At first sight, equations seem to be a redundant feature, at
least in the presence of intermediate predicates.  Indeed,
instead of using an equation $e_1 = e_2$ as a subgoal, we can introduce an
auxiliary relation $T(e_1,e_2)$, and replace the equation by the
subgoal $T(e_1,e_1)$. (Our notation here is not precise but hopefully
enough to convey the idea).  With negated equations and
recursion, however, this simple trick does not work as it
violates stratification.  We still show, however, that
equations are redundant in the presence of both intermediate predicates
and negation.
\item
In the absence of intermediate predicates,
however, equations are a primitive feature.  Indeed, the
``only $a$'s'' query mentioned above, easily
expressed with an equation, is not expressible in the absence of
intermediate predicates.
\item
One can also, conversely, simulate intermediate predicates using
equations: a simple folding transformation works in the absence
of negation and recursion.
In the presence of negation or recursion, however, intermediate
predicates do add power.  This is fairly easy to show for recursion:
the squaring query ``for every path $p$ in the input, output
$a^{n^2}$, where $n$ is the length of $p$'' requires an intermediate
predicate in which the output can be constructed recursively.
In the presence of negation, the primitivity of intermediate
predicates can be seen to follow from the corresponding result
for classical Datalog (by quantifier alternation).  Some work has
still to be done, however, since the classical proof has to be
extended to account for path expressions and equations.
\item
It will not surprise the reader
that recursion is primitive in Sequence Datalog.  This can be
seen in many ways; probably the easiest is to use the above
squaring query, and to observe that
without recursion, the length of output sequences is at most
linear in the length of input sequences.  Another proof, that
also works for boolean queries, is by reduction to the classical
inexpressibility of graph connectivity in first-order logic.
As in the previous paragraph, the
reduction must account for the use of path expressions and
equations.
\item
A classical fact is that nonrecursive Datalog with stratified
negation is equivalent to the relational algebra.  We extend
the standard relational algebra by allowing path expressions in
selection and projection, and adding operators for unpacking and
for subsequences.  We obtain
a language equivalent to nonrecursive Sequence Datalog.
\end{enumerate}

\begin{figure}
\centering
\scalebox{0.82}{\includegraphics{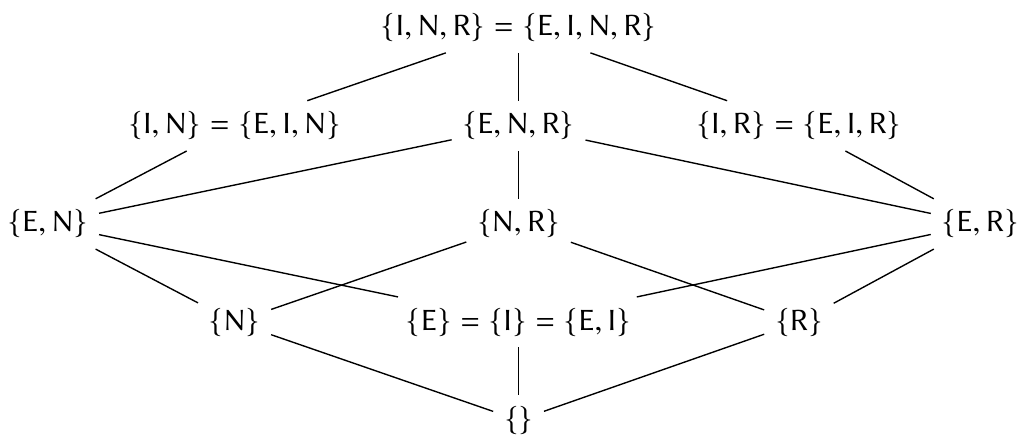}}
  \caption{Relative expressiveness of
  the different sets of Sequence Datalog features
  (Negation, Equations, Intermediate predicates, and Recursion;
  features Arity and Packing will turn out to be entirely
  redundant).
  An ascending path denotes subsumption; absence of such a path denotes
  non-subsumption.}
  \label{fig:eqvClass}
\end{figure}

Our results allow us to completely classify the sixteen possible
Sequence Datalog fragments in a Hasse diagram with respect to
their expressive power, as shown in Figure~\ref{fig:eqvClass}.
Some fragments are equivalent, as shown; also, the features for
packing and higher-arity intermedicate predicates are omitted,
since they are redundant independently of the presence or absence
of other features.

This paper is organized as follows.  In Section~\ref{secseq} we
define the sequence database model and the syntax and semantics
of Sequence Datalog. In Section~\ref{secquery} we introduce the
language features and rigorously define what we mean by one
fragment (set of features) being subsumed in expressive power by
another fragment.  Section~\ref{secred} presents our redundancy
(expressibility) results, and Section~\ref{secprim} presents our
primitivity (inexpressibility) results.
The
Hasse diagram of Figure~\ref{fig:eqvClass} is assembled in
Section~\ref{sechasse}.
Section~\ref{secalg} presents the relational algebra for sequence databases.
We conclude in Section~\ref{seconc},
where we also discuss additional related work.

\section{Sequence databases and Sequence Datalog} \label{secseq}

In this section we formally define the sequence database model
and the syntax and semantics of Sequence Datalog.  We do assume
some familiarity with the basic notions of classical Datalog
\cite{ahv_book}.

\subsection{Data model for sequence databases} \label{secmodel}

A \emph{schema}
$\sch$ is a finite set of relation names, each name with an associated
\emph{arity} (a natural number).  We fix a countably infinite universe
$\dom$ of atomic data elements, called \emph{atomic values}.  The sets of
\emph{packed values}, \emph{values}, and \emph{paths} are defined as the
smallest sets satifying the following:
\begin{enumerate}
  \item Every atomic value is a value.
  \item Every finite sequence of values is a path.  The empty
  path is denoted by $\emp$.

When writing down paths, we will separate the
elements by dots, where the $\conc$ symbol also serves as the
usual symbol for concatenation.  Recall that concatenation is
associative.
  \item If $p$ is a path, then $\pack{p}$ is a packed value.
  \item Every packed value is a value.
\end{enumerate}
The set of all paths is denoted by $\paths$.

For example, if $a$, $b$ and $c$ are atomic values, then $a \cdot
b \cdot a$ is a path; $\langle a \cdot b \cdot a\rangle$ is a
packed value; and $c \cdot \langle a \cdot b \cdot a \rangle$ is
again a path.

An \emph{instance} $\inst$ of a schema $\sch$ is a function that assigns
to each relation name $R \in \sch$ a finite $n$-ary relation on
$\Pi$, with $n$ the arity of $R$.

It is natural to identify a value $v$ with the one-length
sequence $v$. Thus values, in particular atomic values,
are also paths.  Hence, classical
relational database instances are a special case of instances as
defined here.  We refer to such instances as \emph{classical}.
So, in a classical instance, each relation name $R$ is assigned a
finite relation on $\dom$.

\subsection{Syntax of Sequence Datalog}

We assume disjoint supplies of
\emph{atomic variables} (ranging over atomic values)
and \emph{path variables} (ranging over paths).
The set of all variables is also disjoint from $\dom$.
We indicate atomic variables as $\avar x$ and path variables as $\pvar x$.
\emph{Path expressions} are defined just like paths, but with
variables added in.
Formally, we define the set of path expressions to be
the smallest set such that:
\begin{enumerate}
  \item Every atomic value is a path expression;
  \item Every variable is a path expression;
  \item If $e$ is a path expression, then $\pack{e}$ is a path expression;
  \item Every finite sequence of path expressions is a path expression.
\end{enumerate}

A \emph{predicate} is an expression of the form $P(e_1, \ldots, e_n)$,
with $P$ a relation name of arity $n$, and each $e_i$ a path expression.
We call $e_i$ the $i$th component of the predicate.
An \emph{equation} is an expression of the form $e_1 = e_2$,
with $e_1$ and $e_2$ path expressions.

Many of the following definitions adapt well-known Datalog
notions to our data model.

An \emph{atom} is a predicate or an equation.
A \emph{negated atom} is an expression of the form
$\lnot A$ with $A$ an atom.  We write a negated equation $\neg
e_1=e_2$ also as a nonequality $e_1 \neq e_2$.
A \emph{literal} is an atom (also called a positive literal)
or a negated atom (a negative literal).
A \emph{body} is a finite set of literals
(possibly empty).
A \emph{rule} is an expression of the form $H \leftarrow B$,
where $H$ is a predicate, called the \emph{head} of the rule,
and $B$ is a body.

We define the \emph{limited variables} of a rule as the smallest set
such that:
\begin{enumerate}
\item every variable occurring in a positive predicate in the
body is
limited; and
\item if all variables occurring in one of the sides of a positive
equation in the body are limited, then all variables occurring in
the other side are also limited.
\end{enumerate}
A rule is called \emph{safe} if all variables occurring in the rule are
limited.

Finally, a \emph{program} is a finite sequence of \emph{strata},
which are finite sets of safe rules, so that use of negation in
the program is stratified.  Recall that stratified negation means
that when a negated predicate $\neg P(e_1,\dots,e_n)$ occurs in
some stratum, then no rule in that stratum or later strata can
use $P$ in the head predicate.

Note that classical Datalog programs with stratified negation are
a special case of our notion of programs, where the only path
expressions used are atomic values or atomic variables.

\begin{example} \label{nfa}
An NFA can be represented by a unary relation $N$ (initial
states), a ternary relation $D$ (transitions), and a unary
relation $F$ (final states).  These would be classical relations.
Now consider a unary relation $R$ containing paths without
packing, i.e., strings of atomic values.
Then the following program, consisting of a single
stratum, computes in relation $A$ the strings from $R$ that are
accepted by the NFA\@.  Recall that atomic variables are prefixed
with @, and path variables with \$.
\begin{lstlisting}[style=Prolog-pygsty,escapeinside={[]}]
S(@q/-x, !) [] :- R(-x), N(@q).
S(@q2/-y, -z/@a) :- S(@q1/@a/-y, -z), D(@q1, @a, @q2).
A(-x) :- S(@q,-x), F(@q).
\end{lstlisting}
\end{example}

\begin{example} \label{cool}
Consider unary relations $R$ and $S$.
The following program, again in a single stratum, uses packing
and nonequalities to
check whether there are at least three different occurrences of a
string from $S$ as a substring in strings from $R$.
The boolean result is computed in the nullary relation $A$.
\begin{lstlisting}[style=Prolog-pygsty]
T(-u/<-s>/-v) :- R(-u/-s/-v), S(-s).
A :- T(-x),T(-y),T(-z), -x;-y, -x;-z, -y;-z.
\end{lstlisting}
\end{example}

\subsection{Semantics}

We have defined the notion of instance as an assignment of
relations over $\paths$ to relation names.
A convenient equivalent view of instances is as sets of facts.
A \emph{fact} is an expression of the form $R(p_1, \ldots, p_n)$
with $R$ a relation name of arity $n$, and each $p_i$ a path.
An instance $\inst$ of a schema $\sch$ is viewed as the set
of facts $\{R(p_1, \ldots, p_n) \mid
R \in \sch \text{ and } (p_1, \ldots, p_n) \in \inst(R)\}$.

A \emph{valuation} $\nu$ is a function that maps atomic variables
to atomic values and path variables to paths.
We say that $\nu$ is \emph{appropriate} for a syntactical construct
(such as a path expression, a literal, or a rule) if $\nu$ is
defined on all variables in that syntactical construct.
We can apply an appropriate valuation $\nu$ to a path expression
$e$ by substituting each variable in $e$ by its image under $\nu$
and obtain the path $\nu(e)$.
Likewise, we can apply an appropriate valuation to a predicate
and obtain a fact.

Let $L$ be a literal, $\nu$ a valuation appropriate for $L$,
and $\inst$ an instance.
The definition of when $\inst, \nu$ satisfies $L$ is
as expected: if $L$ is a predicate, then the fact $\nu(L)$
must be in $\inst$; if $L$ is an equation $e_1 = e_2$,
then $\nu(e_1)$ and $\nu(e_2)$ must be the same value.
If $L$ is a negated atom $\lnot A$, then $\inst, \nu$ must
not satisfy $A$.

A body $B$ is satisfied by $\inst, \nu$ if all its literals are.
Now a rule $r = H \leftarrow B$ is satisfied in $\inst$
if for every valuation $\nu$ appropriate for $r$ such that
$\inst, \nu$ satisfies $B$, also $\inst, \nu$ satisfies $H$.

The relation names occurring in a program are traditionally
divided into
EDB and IDB relation names.  The IDB relation names are the relation
names used in the head of some rules; the other relation
names are the EDB relation names.
Given a schema $\sch$, a program is said to be over $\sch$
if all its EDB relation names belong to $\sch$, and its
IDB relation names do not.
Now the semantics of programs is
defined as usual. A program is called \emph{semipositive}
if negated predicates only use EDB relation names.
We first apply the first stratum, which is semipositive, and
then apply each subsequent stratum as a semipositive program
to the result of the preceding strata.
So we only need to give semantics for semipositive programs.
Let $\pr P$ be a semipositive program over $\sch$, and let $\inst$
be an instance over $\sch$.
Let $\sch'$ be the set of IDB relation names of $\pr P$.
Then $\pr{P}(\inst)$ is the smallest instance over $\sch \cup \sch'$ that
satisfies all the rules of $\pr P$, and that agrees with $\inst$ on $\sch$.

Due to recursion, for some programs or instances, $\pr{P}(\inst)$
may be undefined, since instances are required to be finite.  We
also say in this case that $\pr P$ does not \emph{terminate} on
$\inst$.  If, in the course of evaluating a program $\pr P$ with
several strata on an instance $\inst$, one of the strata does not
terminate, we agree that the entire program $\pr P$ is undefined
on $\inst$.  As mentioned in the Introduction, Bonner and Mecca
have done substantial work on the question of guaranteeing
termination for Sequence Datalog programs.  In this paper, we
only consider programs that always terminate.

\begin{example}
The program from Example~\ref{nfa}, while recursive,
is guaranteed to terminate on every instance.  In contrast, the following
two-rule program will not terminate on any instance:
\begin{lstlisting}[style=Prolog-pygsty]
T(a).
T(a/-x) :- T(-x).
\end{lstlisting}
\end{example}

\section{Features, fragments, and queries} \label{secquery}

In this paper, we consider six possible features that a program
may use, each identified by a letter, spelled out as follows.

\begin{description}
\item[Arity]
A program \emph{uses arity} (has the
$\A$-feature) if it contains at least one predicate of arity greater than one.
\item[Recursion]
A program \emph{uses recursion} (has the $\R$-feature)
if there is a cycle in its dependency graph.\footnote{The nodes
of this graph are the IDB relation names, and there is an edge
from $R_1$ to $R_2$ if $R_2$ occurs in the body of a rule with
$R_1$ in its head predicate.}
\item[Equations]
A program \emph{uses equations} (has the
$\E$-feature) if it contains at least one equation in some rule.
\item[Negation]
A program \emph{uses negation} (has the
$\N$-feature) if it contains at least one negated atom in some rule.
\item[Packing] A program \emph{uses packing} (has the
$\P$-feature) if a path expression of the form $\pack{e}$ occurs in some rule.
\item[Intermediate predicates] A program \emph{uses
intermediate predicates} (has the $\I$-feature) if it involves at least two
different IDB relation names.
\end{description}

Let $\feat = \{\A,\E,\I,\N, \allowbreak \P,\R\}$ be the set of
all features.  A subset of $\feat$ is called \emph{a fragment}.
A program $\pr P$ is said to belong to a fragment $F$ if it uses
only features from $F$.

\begin{example} \label{exfragment}
The following program belongs to fragment $\{\E\}$.  It computes,
in relation $S$, all paths from $R$ that consist exclusively of
$a$'s.
\begin{lstlisting}[style=Prolog-pygsty]
  S(-x) :- R(-x),  a/-x=-x/a.
\end{lstlisting}
The following program does the same, but belongs to fragment
$\{\A,\I,\R\}$:
\begin{lstlisting}[style=Prolog-pygsty]
  T(-x, -x) :- R(-x).
  T(-x, -y) :- T(-x, -y/a).
  S(-x) :- T(-x, !).
\end{lstlisting}
\end{example}

\subsection{Queries and subsumption among fragments}

Our goal is to compare the different fragments with respect to
their power in expressing queries.  Our methodology is to do this
relative to a baseline class of queries that do not presuppose
any feature to begin with.  We next define these queries formally.

We call a schema \emph{monadic} if each of its relation names has
arity zero or one.  Also, we call an instance \emph{flat} if it
contains no occurrences of packed values.

Given a monadic schema $\sch$ and relation name $S \not \in \sch$
of arity at most one, a \emph{query} from $\sch$ to $S$ is a
total mapping from flat instances over $\sch$ to flat instances
over $\{S\}$.
A program $\pr P$ is said to \emph{compute} such a query if
\begin{enumerate}
  \item $\pr P$ is over $\sch$;
  \item $\pr P$ terminates on every flat instance of $\sch$;
  \item $S$ is an IDB relation of $\pr P$; and
  \item $\pr{P}(\inst)(S)$ equals $Q(\inst)$ for every flat
  instance $\inst$ of $\sch$.
\end{enumerate}

We now say that fragment $F_1$ is \emph{subsumed} by fragment
$F_2$, denoted by $F_1 \subs  F_2$, if every query computable by
a program in $F_1$ is also computable by a program in $F_2$.
Note that it is possible, for different $F_1$ and $F_2$,
that $F_1 \subs F_2$ and $F_2 \subs F_1$.  Such two fragments are
equivalent in expressive power.  There will turn out to be
11 equivalence classes; in Section~\ref{sechasse} we will
give a theorem that will characterize the subsumption relation
as shown in Figure~\ref{fig:eqvClass}.

\subsection{Redundancy and primitivity}

We will explore the subsumption relation by investigating the
redundancy or primitivity of the different features with respect
to other features.  A feature might be redundant in an absolute
sense, in that it can be dropped from any fragment without
decrease in expressive power.
This is a very strong notion of redundancy, and we
cannot expect it to hold for most features.  Yet a more relative
notion of redundancy may hold, meaning that some feature does
not contribute to expressive power, on condition that some other
features are already present, or are absent.  This leads to the
following notions.

\begin{definition}[Redundancy]
  Let $X$ be a feature and let $Y$ and $Z$ be sets of features.
  \begin{itemize}
    \item $X$ is redundant if $F \subs F-\{X\}$ for every fragment $F$.
    \item $X$ is redundant in the presence of $Y$ if $F \subs F-\{X\}$ for
    every fragment $F$ such that $Y \subseteq F$.
    \item $X$ is redundant in the absence of $Z$ if $F \subs F-\{X\}$ for
    every fragment $F$ such that $Z$ is disjoint from $F$.
    \item $X$ is redundant in the presence of $Y$ and absence of $Z$
    if $F \subs F-\{X\}$ for every fragment $F$ such that $Y \subseteq F$ and
    $Z$ is disjoint from $F$.
  \end{itemize}
\end{definition}

Similarly, but conversely,
a feature might be primitive in an absolute sense, in
that dropping it from a fragment always strictly decreases the
expressive power.  Then again, for other features only
more relative notions of primitivity may hold.

\begin{definition}[Primitivity]
  Let $X$ be a feature and let $Y$ and $Z$ be sets of features.
  Recall that $\feat$ is the set of all features.
  \begin{itemize}
    \item $X$ is primitive if $\{X\} \nsubs \feat - \{X\}$.

    \item $X$ is primitive in the presence of $Y$ if
    $\{X\} \cup Y \nsubs \feat - \{X\}$.

    \item $X$ is primitive in the absence of $Z$ if
    $\{X\} \nsubs \feat - (\{X\} \cup Z)$.
  \end{itemize}
\end{definition}

\section{Expressibility results} \label{secred}

In this section we show various expressibility results that lead
to absolute or relative redundancy results for various features.

\subsection{Arity}

Using a simple encoding trick we can see that arity is redundant.
Indeed, let $a$ and $b$ be two different atomic values.  For any
paths $s_1$, $s_2$, $s_1'$ and $s_2'$, we have the following:

\begin{lemma} \label{lemenc}
$(s_1,s_2)=(s_1',s_2')$ if and only if
$$ \enc{s_1}{s_2} = \enc{s_1'}{s_2'}. $$
\end{lemma}
\begin{proof}
  The if-direction is trivial.
  For the only-if direction, we consider $\enc{s_1}{s_2} = \enc{s_1'}{s_2'}$
  and we observe that $a$ appears in the middle of both sequences.
  Hence,
  \begin{itemize}
    \item[(a)] $s_1 \conc a \conc s_2 = s_1' \conc a \conc s_2'$ and
    \item[(b)] $s_1 \conc b \conc s_2 = s_1' \conc b \conc s_2'$.
  \end{itemize}

For the sake of contradiction,
let us assume $|s_1| < |s_1'|$. Then $s_1'= s_1 \conc x$ for a
nonempty sequence $x$.  Thus, equation (a) can be rewritten as
$s_1 \conc a \conc s_2 = s_1 \conc x \conc a \conc s_2'$, which
simplifies to $a \conc s_2 = x \conc a \conc s_2'$.  Hence the
sequence $x$ must start with $a$.  In the same way, however, we can deduce
from (b) that $x$ must start with $b$.  Hence, the assumption we
made is false.

Analogously, $|s_1| > |s_1'|$ can be seen to be false as well, so we know that
$|s_1| = |s_1'|$.  Then clearly $|s_2| = |s_2'|$
    as well.
    Hence, from (a) and (b) we get that $s_1 = s_1'$ and $s_2 = s_2'$.
\end{proof}

Using this encoding, arities higher than one can be reduced by
one.  Since we can do this repeatedly, we obtain:

\begin{theorem} \label{thm:arity}
Arity is redundant.
\end{theorem}

\begin{example}
Consider the following program which computes in $S$
  the reversals of the paths in $R$:
\begin{lstlisting}[style=Prolog-pygsty,escapeinside={[]}]
  T(-x, !) [] :- R(-x).
  T(-x, -y/@u) :- T(-x/@u, -y).
  S(-x) :- T(!, -x).
\end{lstlisting}
The same query can be expressed without arity as follows:
\begin{lstlisting}[style=Prolog-pygsty,label=lst:exar2]
  T(-x/a/a/-x/b) :- R(-x).
  T(-x/a/-y/@u/a/-x/b/-y/@u) :- T(-x/@u/a/-y/a/-x/@u/b/-y).
  S(-x) :- T(a/-x/a/b/-x).
\end{lstlisting}
\end{example}

\subsection{Equations}

In the presence of $\I$ and $\A$, positive equations are readily
seen to be redundant, by introducing an auxiliary intermediate
predicate in the program. We only give an example:

\begin{example}
Recall the program from Example~\ref{exfragment}:
  \begin{lstlisting}[style=Prolog-pygsty]
  S(-x) :- R(-x),  a/-x=-x/a.
  \end{lstlisting}
The same query can be computed without equations as follows:
  \begin{lstlisting}[style=Prolog-pygsty]
  T(a/-x, -x) :- R(-x).
  S(-x) :- T(-x/a, -x).
  \end{lstlisting}
\end{example}

This simple method works only in the absence of negation,
because, when applied to a negated equation in a rule that
belongs to a recursive stratum, stratification is violated.
However, negated equations can be handled by another method:

\begin{lemma}
$\E$ is redundant in the presence of $\I$, $\A$ and
$\N$.
\end{lemma}
\begin{proof}
Positive equations can be handled as above.  For each stratum
$\stratum$ that contains negated equations, we insert a new
stratum $\stratum'$, right before $\stratum$, consisting of the
following rules.  Let $\rho$ be a renaming that maps each head
relation name in $\stratum$ to a fresh relation name; relation
names that occur only in bodies in $\stratum$ are mapped to
themselves by $\rho$.

For each rule $H \leftarrow B$ in $\stratum$ without negated
  equations, we add the rule $\rho(H) \gets \rho(B)$ to $\stratum'$.

For each rule $r: H \leftarrow B \land e_1 \neq e_1' \land \ldots
\land e_n \neq e_n'$ in $\stratum$ with $n$ negated equations, we
again add $\rho(H) \gets \rho(B)$ to $\stratum'$.  Moreover,
using a fresh relation name $T$, we add the following $n$ rules
for $i=1,\dots,n$:
\begin{tabbing}
$T(v_1, \ldots, v_m) \leftarrow \rho(B) \land e_i = e_i'$
\end{tabbing}
Here, the $v$'s are all variables appearing in $B$.

Finally, in $\Delta$, we replace $r$ by the following rule:
\begin{tabbing}
$H \leftarrow B \land \lnot T(v_1, \ldots, v_m)$.
\end{tabbing}
\end{proof}

\begin{example}
The following program retrieves in $S$ those paths from $R$
that can be written as $a_1 \cdots a_n b_n \cdots b_1$ with $a_i
\neq b_i$ for $i=1,\dots,n$:
  \begin{lstlisting}[style=Prolog-pygsty,label=lst:exeq3]
  U(-x, -x) :- R(-x).
  U(-x, -y) :- U(-x, @a/-y/@b), @a;@b.
  S(-x) :- U(-x, !).
  \end{lstlisting}
Applying the method to eliminate negated equations, we obtain:
  \begin{lstlisting}[style=Prolog-pygsty,label=lst:exeq4]
  U1(-x, -x) :- R(-x).
  U1(-x, -y) :- U1(-x, @a/-y/@b).
  T(-x, -y, @a, @b) :- U1(-x, @a/-y/@b), @a=@b.
  S1(-x) :- U1(-x, !).
  U(-x, -x) :- R(-x).
  U(-x, -y) :- U(-x, @a/-y/@b), + T(-x, -y, @a, @b).
  S(-x) :- U(-x, !).
  \end{lstlisting}
\end{example}

From the above we conclude that $\E$ is redundant in the presence
of $\I$ and $\A$.  Since we already know that arity is redundant,
we obtain:

\begin{theorem}\label{thm:equatons}
  $\E$ is redundant in the presence of $\I$.
\end{theorem}

\subsection{Packing}

In this section we show that packing is redundant.
The main task will be to eliminate packing from equations
in nonrecursive programs.  We will follow the following strategy
to achieve this task:

\begin{enumerate}
  \item In Section~\ref{secpure} we show how to
  eliminate all variables that can hold values with
  packing. We will call such variables \emph{impure}.
  The elimination is achieved by ``solving''
equations involving impure variables.
\item
Thereto, we will extend a known method for solving word
  equations.  We begin by recalling this method in Section~\ref{secword}.
In Section~\ref{secwp} we
  present the extension to path expressions.
  \item When all variables are pure,
equations involving packing can only be satisfiable if the two sides
have a similar ``shape'', called \emph{packing structure}.
We formalize this in Section~\ref{secstruct}.
\end{enumerate}
The main result concerning packing is then proven
in Section~\ref{secpmain}.

\subsubsection{Solving equations} \label{secword}

Consider an equation $e_1 = e_2$ and let $X$ be the set of
variables occurring in the equation.  A valuation $\nu$ on $X$ is
called a \emph{solution} if $\nu(e_1)$ and $\nu(e_2)$ are the
same path.  The set of solutions is typically infinite, so we
would like a way to represent this set in a finite manner.

Thereto one can use variable substitutions: partial functions
that map variables to path expressions over $X$.  Such a variable
substitution $\rho$ is called a \emph{symbolic solution} to the
equation if $\rho(e_1)$ and $\rho(e_2)$ are the same path
expression.  Every symbolic solution $\rho$ represents a set of
solutions $$ [\rho] := \{\nu \circ \rho \mid \text{$\nu$ a
valuation on $X$}\}. $$ A set $R$ of symbolic solutions is
called \emph{complete} if $\bigcup \{[\rho] \mid \rho \in R\}$
yields the complete set of solutions to the equation.

The classical setting of \emph{word equations}
\cite{abdulrab1989solving} can be seen as a
special case of the situation just described.  A word equation
corresponds to the case where $e_1$ and $e_2$ contain no
packing, and no atomic variables, i.e., all variables are path
variables.

Plotkin's ``pig-pug'' procedure for associative unification
\cite{plotkin1972building} generates a complete set of symbolic
solutions to any word equation.  However, not every word equation
admits a \emph{finite} complete set of symbolic solutions; a
simple example is our familiar equation $\$x \cdot a = a \cdot
\$x$.  Hence, in general, the procedure may not
terminate.\footnote{The reader may be interested to know that
other
means of finite representation (different from a finite set of
substitutions) have been discovered, that work for arbitrary word
equations \cite{plandowski_fswep}.} Nevertheless, pig-pug is
guaranteed to terminate on ``one-sided nonlinear'' equations
\cite{duran2018associative}.  These are word equations where all
variables that occur more than once in the equation, only occur
in one side of the equation.

We briefly review the pig-pug procedure.
The procedure constructs a search tree whose nodes are labeled with word
equations; the root is labeled with the original word
equation. For each node we generate children according
to a rewriting relation, $\Rightarrow$, on word equations.  Specifically, we
have the following rewrite rules:
\begin{enumerate}
  \item Cancellation rule:
$(x \conc w_1 \weq x \conc w_2) \Rightarrow (w_1 \weq w_2)$,
for $x \in \dom \cup X$.

  \item Main rules: each one of the rules is associated with a substitution,
  $\rho$.  Let $\bf x$ and $\bf y$ be distinct variables and let
  $a$ be an atomic value.
  \begin{enumerate}
    \item $({\bf x} \conc w_1 \weq {\bf y} \conc w_2) \Rightarrow
    ({\bf x} \conc \rho(w_1) \weq \rho(w_2))$ with
    $\rho({\bf x}) = {\bf y \conc x}$

    \item $({\bf x} \conc w_1 \weq {\bf y} \conc w_2) \Rightarrow
    (\rho(w_1) \weq \rho(w_2))$ with $\rho({\bf x}) = {\bf y}$

    \item $({\bf x} \conc w_1 \weq {\bf y} \conc w_2) \Rightarrow
    (\rho(w_1) \weq {\bf y} \conc \rho(w_2))$ with
    $\rho({\bf y}) = {\bf x \conc y}$

    \item $({\bf x} \conc w_1 \weq a \conc w_2) \Rightarrow
    ({\bf x} \conc \rho(w_1) \weq \rho(w_2))$ with
    $\rho({\bf x}) = a {\bf \conc x}$

    \item $({\bf x} \conc w_1 \weq a \conc w_2) \Rightarrow
    (\rho(w_1) \weq \rho(w_2))$ with $\rho({\bf x}) = a$

    \item $(a \conc w_1 \weq {\bf y} \conc w_2) \Rightarrow
    (\rho(w_1) \weq {\bf y} \conc \rho(w_2))$ with
    $\rho({\bf y}) = a {\bf \conc y}$

    \item $(a \conc w_1 \weq {\bf y} \conc w_2) \Rightarrow
    (\rho(w_1) \weq \rho(w_2))$ with $\rho({\bf y}) = a$
  \end{enumerate}
\end{enumerate}

When no rule is applicable to an equation, we have reached a leaf node in
the search tree.  There are three possible cases for such a leaf equation:
\begin{enumerate}
  \item $(\emp \weq \emp)$.
  \item $(a \conc w_1) \weq (b \conc w_2)$, for atomic
  values $a \neq b$.
  \item $(\emp \weq w)$ or $(w \weq \emp)$, for nonempty $w$.
\end{enumerate}
The first case is successful, while the other two are not.  Each
path from the root to a leaf node of the form $(\emp \weq \emp)$
yields a symbolic solution, formed by composing the substitutions
given by the rewritings along the path.  When starting from a
one-side nonlinear equation, the tree is finite and we obtain a complete
finite set of symbolic solutions.\footnote{It is
standard in the literature on word equations to consider
only solutions that map variables to nonempty words.  The above
procedure is only complete under that assumption.  However,
allowing the empty word can be easily accommodated.
For any equation $\mathit{eq}$ on a set of variables $X$, and any
subset $Y$ of $X$, let $\mathit{eq}_Y$ be the equation obtained
from $\mathit{eq}$ by replacing the variables in $Y$ by the empty
word.  Let $R_Y$ be a complete set of symbolic solutions for
$\mathit{eq}_Y$ where we extend each substitution to $X$ by
mapping every variable from $Y$ to the empty word.  Then the
union of the $R_Y$ is a complete set of symbolic solutions for
$\mathit{eq}$, allowing the empty word.  If $\mathit{eq}$ is
one-sided nonlinear, then $\mathit{eq}_Y$ is too.  This remark
equally applies to the extension to path expressions presented in
Section~\ref{secwp}.}

\input{packing-example}

\subsubsection{Extension to path expressions} \label{secwp}

Our equations differ from word equations in that path expressions
can involve packing as well as atomic variables.  To this end, we
extend the rewriting system as follows.

\begin{enumerate}[label=(\alph*)]
  \setcounter{enumi}{7}
  \item Given an equation of the form $(\avar x \conc w_1 \weq \avar y
  \conc w_2)$, the only possibility is for $\avar x$ and $\avar y$ to be
  the same.  Thus we add the rule $(\avar x \conc w_1 \weq \avar y \conc w_2)
  \Rightarrow (\rho(w_1) \weq \rho(w_2))$ with $\rho(\avar x) = \avar y$.

  \item An equation of the form $(\avar x \conc w_1 \weq \pvar y
  \conc w_2)$ is not very different from the case where we have
  a constant instead of $\avar x$.  Thus, we add two
  rules similar to rules (f) and (g):
    \begin{itemize}
      \item $(\avar x \conc w_1 \weq \pvar y \conc w_2) \Rightarrow
      (\rho(w_1) \weq \pvar y \conc \rho(w_2))$ with $\rho(\pvar y) =
      \avar x \conc \pvar y$

      \item $(\avar x \conc w_1 \weq \pvar y \conc w_2) \Rightarrow
      (\rho(w_1) \weq \rho(w_2))$ with $\rho(\pvar y) = \avar x$
    \end{itemize}

  \item Analogously, we add rules similar to rules (d) and (e):
    \begin{itemize}
      \item $(\pvar x \conc w_1 \weq \avar y \conc w_2) \Rightarrow
      (\pvar x \conc \rho(w_1) \weq \rho(w_2))$ with $\rho(\pvar x) =
      \avar y \conc \pvar x$

      \item $(\pvar x \conc w_1 \weq \avar y \conc w_2) \Rightarrow
      (\rho(w_1) \weq \rho(w_2))$ with $\rho(\pvar x) = \avar y$
    \end{itemize}

  \item Given an equation of the form $(\pack{w_1} \conc w_2 \weq
  \pack{w_3} \conc w_4)$, we work inductively
  and solve the equation $w_1 \weq w_3$ first.  Assuming we can
  find a finite complete set $R$ of symbolic solutions for this
  equation, we then add the rules
  $(\pack{w_1} \conc w_2 \weq \pack{w_3}
  \conc w_4) \Rightarrow (\rho(w_2) \weq \rho(w_4))$ for $\rho
  \in R$.

  \item An equation of the form
  $(\pack{w_1} \conc w_2 \weq \pvar y \conc w_3)$
is again not very different from the case where we have
  a constant instead of $\pack{w_1}$.  Thus, we add two
  rules similar to rules (f) and (g):
    \begin{itemize}
      \item $(\pack{w_1} \conc w_2 \weq \pvar y \conc w_3) \Rightarrow
      (\rho(w_2) \weq \pvar y \conc \rho(w_3))$ with $\rho(\pvar y) =
      \pack{w_1} \conc \pvar y$

      \item $(\pack{w_1} \conc w_2 \weq \pvar y \conc w_3) \Rightarrow
      (\rho(w_2) \weq \rho(w_3))$ with $\rho(\pvar y) = \pack{w_1}$
    \end{itemize}

  \item Analogously, we again add rules similar to rules (d) and
  (e):
    \begin{itemize}
      \item $(\pvar x \conc w_1 \weq \pack{w_2} \conc w_3) \Rightarrow
      (\pvar x \conc \rho(w_1) \weq \rho(w_3))$ with $\rho(\pvar x) =
      \pack{w_2} \conc \pvar x$

      \item $(\pvar x \conc w_1 \weq \pack{w_2} \conc w_3) \Rightarrow
      (\rho(w_1) \weq \rho(w_3))$ with $\rho(\pvar x) = \pack{w_2}$
    \end{itemize}
\end{enumerate}

Furthermore, we now have extra non-successful cases for leaf
equations, namely all equations of the form
$(\avar x \conc w_1 \weq \pack{w_2} \conc w_3)$ or
$(\pack{w_2} \conc w_3 \weq \avar y \conc w_1)$.

Extending known arguments, one can see that on any one-sided
nonlinear equation, our extended rewriting system terminates and yields a
finite complete set of symbolic solutions.

\begin{example}
  Figure~\ref{fig:packingex1} shows a DAG representation of the
  search tree for the
  equation $\pvar x \conc \pack{\avar y \conc \pvar z}
  \conc \avar w = \pvar u \conc \pvar v \conc \pvar u$.
There are four successful branches, so the following
substitutions comprise a complete set of symbolic solutions:
  \begin{align*}
& \{\pvar x \mapsto \avar w,\,
    \pvar u \mapsto \avar w,\,
    \pvar v \mapsto \pack{\avar y \conc \pvar z}\}
\\
& \{\pvar x \mapsto \avar w \conc \pvar x,\,
    \pvar v \mapsto \pvar x \conc \pack{\avar y \conc \pvar z},\,
    \pvar u \mapsto \avar w\}
\\
& \{\pvar x \mapsto
    \pack{\avar y \conc \pvar z} \conc \avar w \conc \pvar v,\,
    \pvar u \mapsto \pack{\avar y \conc \pvar z} \conc \avar w\}
\\
& \{\pvar x \mapsto \pvar x \conc \pack{\avar y \conc \pvar z}
    \conc \avar w \conc \pvar v \conc \pvar x,\,
    \pvar u \mapsto \pvar x \conc \pack{\avar y \conc \pvar z}
    \conc \avar w\}
  \end{align*}
\end{example}

\subsubsection{Pure variables and pure equations} \label{secpure}

We introduce a syntactic ``purity check'' on variables, that
guarantees that they can only take values that do not contain
packed values.  Since later we will work stratum per stratum, it
is sufficient in what follows to focus on semipositive,
nonrecursive programs with only one IDB relation name.

Consider a rule in such a program.  When a variable appears in
some positive EDB predicate, we call the variable a
\emph{source variable} of the rule.  Now we inductively define
a variable in the rule to be \emph{pure} if
\begin{enumerate}
  \item it is a source variable (since we focus on flat input instances); or
  \item it appears in one side of a positive equation, such that
  \begin{itemize}
    \item all the variables in the other side of the equation are pure, and
    \item the other side of the equation has no packing.
  \end{itemize}
\end{enumerate}

By leveraging associative unification, we are going to show that
we can always eliminate impure variables.  The method is based on
a division of the positive equations of a rule into three categories:
\begin{description}
  \item[Pure equations] involve only pure variables.
  \item[Half-pure equations] have all variables in
  one side pure, and at least one of the variables in the other side is
  impure.
  \item[Fully impure equations] have impure
  variables in both sides.
\end{description}

\begin{example}
The three equations in the rule $$\ru{S(\pvar x)} {R(\pvar x, \pvar y)
    \land \pack{\pvar x} = \pack{\pvar y} \land a \conc \pvar x = \pv z \land
    \pvar y = \pack{\pvar u}}$$ are pure.
The two equations in the rule $$\ru{S(\pvar x)} {R(\pvar x, \pvar y) \land
    \pack{\pvar y} = \pvar z \land \pack{\pvar x} = \pack{\pvar z}}$$ are
    half-pure.
The equation $\pack{\pvar t} = \pack{\pvar z}$ in the rule
    $$\ru{S(\pvar x)} {R(\pvar x, \pvar y) \land \pack{\pvar t} =
    \pack{\pvar z} \land \pvar z = \pack{\pvar y} \land
    \pvar t = \pack{\pvar x}}$$ is fully impure.
\end{example}

It is instructive to compare the notion of pure variable with that of limited
variable, used to define the notion of safe rule.  Indeed, the set of limited
variables can be equivalently defined as follows, where we only change the base
case of the induction to immediately include all pure variables:
\begin{itemize}
  \item Every pure variable is limited; and
  \item If all the variables occurring in one side of the sides of an
  equation in the rule are limited, then all the variables occurring in
  the other side are also limited.
\end{itemize}

Therefore, if there is at least one impure variable in a safe
rule, then there must be at least one half-pure equation in the
rule.  In other words, it is not possible for a rule to have
fully impure equations without having half-pure ones.

\begin{lemma}\label{lem:non-pure-eq}
  Let $r$ be a rule in a semi-positive, nonrecursive program
  $\pr P$ with only one IDB relation name.
Then there exists
  a finite set of rules, equivalent to $r$ on flat instances,
  in which all positive equations are pure.
\end{lemma}
\begin{proof}
By induction on the number of half-pure equations.
Let $r: H \leftarrow B \land e_1=e_2$, where $e_1=e_2$ is half-pure with
$e_1$ the pure side and $e_2$ the impure side.
Let $u_1, \ldots, u_n$ be the list of all \emph{occurrences} of variables
  in $e_1$. Let $v_1, \ldots, v_n$ be $n$ fresh variables, and let
  $e_1'$ be $e_1$ with each $u_i$ replaced by $v_i$.
Now replace $e_1=e_2$ by
  the following conjunction of $n+1$ equations:
  $$ u_1=v_1 \land \ldots \land u_n=v_n \land e_1'=e_2 $$
Here, abusing notation, we use the same notation
$u_i$ for the variable that occurs at $u_i$.

Denote the result of this replacement by $r'$.
The equation $e_1'=e_2$ is one-sided
nonlinear; by Section~\ref{secwp}, there exists a finite complete
set $R$ of solutions.  If we let $r''$ be $r'$ without
$e_1'=e_2$, then clearly $r$ is equivalent to the set of rules
$\{\rho(r'') \mid \rho \in R\}$.  However, some of these rules
may not have strictly less half-pure equations than $r$, which is
necessary for the induction to work.

We can solve this problem as follows.  Call $\rho \in R$
\emph{valid} if it maps variables that are pure in $r''$ to
expressions without packing.  Since all $u_i$ and $v_i$ are pure
in $r''$, the equations $\rho(u_i)=\rho(v_i)$ in $\rho(r'')$ are
all pure, so $\rho(r'')$ does have strictly less half-pure
equations than $r$.

Fortunately, we can restrict attention to the valid $\rho \in R$,
so the induction goes through.  Indeed, following the definition of
pure variable, one can readily verify that for nonvalid $\rho$, the rule
$\rho(r'')$ is unsatisfiable on flat instances.
\end{proof}

\subsubsection{Packing structures} \label{secstruct}

By Lemma~\ref{lem:non-pure-eq},
all positive equations can be taken to be pure.  We now
reduce this further so that all positive equations are free of
packing.  Thereto we introduce
the \emph{packing structure} of a path expression $e$, denoted by
$\ps{e}$, and defined as follows:
\begin{itemize}
  \item $\ps{\emp} = *$.
  \item $\ps{a} =*$, with $a$ a variable or an atomic value.
  \item $\ps{\pack{e}} = * \conc \pack{\ps{e}} \conc *$.
  \item $\ps{e_1 \conc e_2}$ equals $\ps{e_1} \conc \ps{e_2}$, in
  which we replace any consecutive sequence of
stars by a single star.
\end{itemize}

Assume $\ps e$ has $n$ stars.  Then $e$ can be constructed from
$\ps e$ by replacing each star by a unique (possibly empty)
subexpression of $e$.  We call these subexpressions the
components of $e$. Crucially, they do not use packing.

If $e$ does not use packing, $\ps e$ is simply $*$.  If $e$
begins or ends with
packing, or if some packing in $e$ begins or ends with another packing, then
some components will be empty.

\begin{example}
Let $e = \avar a \conc \pack{\pack{\pvar x \conc \pvar y} \conc
\pvar z} \conc \pack{\emp}$.  Then $\ps{e} = * \conc \pack{ *
\conc \pack{*} \conc *} \conc * \conc \pack{*} \conc *$.  The
seven components of $e$ are $\avar{a}$, $\emp$, $\pvar{x} \conc
\pvar{y}$, $\pvar{z}$, $\emp$, $\emp$, and $\emp$.
\end{example}

A pure equation $e_1 = e_2$ can only be satisfiable on flat
instances if $e_1$ and $e_2$ have the same packing structure.
Suppose there are $n$ stars in this packing structure.
Then, the equation can be replaced by the conjunction of $n$
equations, where we equate the corresponding components of $e_1$
and $e_2$.  These equations are still pure, and free of packing.

Moreover, when all positive equations are pure, then all
variables in the rule are pure, since the rule is safe.
Now a negated equation $e_1 \neq e_2$ over pure variables is
equivalent to the disjunction of the nonequalities between the
corresponding components of $e_1$ and $e_2$.  Then the rule can
be replaced by a set of rules, one for each disjunct, and the
component nonequalities are free of packing.  We can repeat this
for all negated equations.

We have arrived at the following:

\begin{lemma} \label{lem:pure-eq}
Let $r$ be a rule in a semi-positive, nonrecursive program
  $\pr P$ with only one IDB relation name.
Then there exists
  a finite set of rules, equivalent to $r$ on flat instances,
  in which all variables are pure, and all
  equations (positive or negated) are free of packing.
\end{lemma}

\subsubsection{Redundancy of packing} \label{secpmain}

We are now ready for the following result.  The proof further
leverages packing structures.

\begin{lemma}\label{lem:packing1}
Packing is redundant in the absence of recursion.
\end{lemma}
\begin{proof}
Consider a query computed by a nonrecursive program $\pr P$.  We
must show that $\pr P$ can be equivalently rewritten without
packing.  If $\pr P$ has only one IDB predicate,
Lemma~\ref{lem:pure-eq} gives us what we want.  Indeed, by the
Lemma, we may assume that equations are already free of
packing.  Now since the input is a flat instance, any positive
(negated) EDB predicate that contains packing may be taken to be
always false (true).  Also, the result of the query is a flat
instance, so IDB predicates containing packing are false as well.
We thus obtain a program free of packing as desired.

When $\pr P$ uses intermediate predicates, the elimination of
packing from IDB predicates requires more work.  Since $\pr P$ is
nonrecursive, we may assume that every stratum involves only one
IDB relation name.  Since arity is redundant, we may assume that
$\pr P$ does not use arity, but feel free to use arity in the
rewriting of $\pr P$.

Let us consider the first stratum.  For every rule, we proceed as
follows.  Let $R(e)$ be the head of the rule.  Let $m$ be the
number of stars in $\ps{e}$ and let $e_1$, \ldots, $e_m$ be the
components of $e$.  Replace the head with $R_{\ps{e}}(e_1,
\ldots, e_m)$ where $R_{\ps{e}}$ is a fresh relation name.

After this step, the rules in the first stratum no longer contain
packing in the head.  Of course, $R$-predicates in rules in later
strata must now be updated to call the new relation names.
So, assume $R(e)$ appears in the body of some later rule $r$.  For each of
   the packing structures $ps$ introduced for $R$, we
make a copy of $r$ in which we replace $R(e)$ by the
   conjunction $R_{ps}(\pvar e_1, \ldots, \pvar e_m) \land e=e'$, where
   \begin{itemize}
     \item $m$ is the number of stars in $ps$;
     \item $\pvar e_1, \ldots, \pvar e_m$ are fresh path
     variables; and
     \item $e'$ is obtained from the packing structure $ps$ by replacing the
     $i$th star by $\pvar{e_i}$, for $i=1,\ldots,m$.
   \end{itemize}
This rewriting introduces equations in later strata, which is
necessary because these later strata have not yet been
purified per Lemma~\ref{lem:pure-eq}.

We do the above for every stratum.  So, stratum by stratum, we
first remove packing from equations, leaving only pure variables
in rules; we replace head predicates;
and rewrite calls to these head predicates in later rules.

After this transformation, packing still appears in negated IDB
predicates, which have been untouched so far.  Fortunately, all
rules have pure variables at this point.  Thus, a literal $\neg
R(e)$, where $\ps e$ matches one of the packing structures of
$R$, say $ps$, with $m$ stars, can now be replaced by $\lnot
R_{ps}(e_1, \ldots, e_m)$, where $e_i$ is the $i$th component of
$e$.  If $\ps e$ does not match any of the packing structures
introduced for $R$, the negative literal is true on flat
instances and can be omitted.

Observing that packing in EDB predicates can be handled as in the
semipositive case, we are done.
\end{proof}

\begin{example}
Rewriting the program from Example~\ref{cool} without packing
yields a program with 28 rules:
\begin{lstlisting}[style=Prolog-pygsty, escapeinside={[]}]
T(-u, -s, -v) :- R(-u/-s/-v), S(-s).
A :- T(-x1,-x2,-x3), T(-y1,-y2,-y3), T(-z1,-z2,-z3),
     -xi;-yi, -xj;-zj, -yk;-zk.
% for i=1,2,3, j=1,2,3, k=1,2,3
  \end{lstlisting}
\end{example}

To get from
Lemma~\ref{lem:packing1} to the following theorem, it remains to
show that packing is redundant in the presence of recursion.
Building on the flat--flat theorem for J-Logic
\cite{hidders2017j, hidders2020j} we can close that gap and we obtain:
\begin{theorem} \label{thm:packing}
Packing is redundant.
\end{theorem}
\begin{proof}
It remains to show that
$\P$ is redundant
  in the presence of $\R$.
Earlier work on J-Logic (flat-flat theorem
  \cite{hidders2017j, hidders2020j}) is readily adapted to
  Sequence Datalog and shows that $\P$ is redundant
  in the presence of $\R$ and $\N$.  The general idea of the rewriting used in
  that proof is as follows:
  \begin{enumerate}
    \item We add a new stratum at the beginning of the program,
    where we preprocess
    the input relations
as follows: every path $k_1 \conc
    k_2 \conc \cdots \conc k_n$ is replaced by its doubled
    version $k_1 \conc k_1 \conc k_2 \conc k_2 \conc \cdots \conc k_n \conc
    k_n$.

    \item We modify the program so that it works with doubled EDB and IDB
    relations.  Packing is simulated using
a technique of simulated delimiters, which relies on the doubled
encoding.
    \item In the last step, we undouble the doubled output.
  \end{enumerate}

Steps 1 and 3 as published introduce negation even if the original program does
not use negation. We next show that this can be avoided. Instead, we introduce
arity, which is harmless as arity is redundant.

We double an EDB relation $R$ into $R'$ as follows:
\begin{lstlisting}[style=Prolog-pygsty,label=lst:encoding,escapeinside={[]}]
  T(!, -x) :- R(-x).
  T(-x/@y/@y, -z) :- T(-x, @y/-z).
  R['](-x) :- T(-x, !).
  \end{lstlisting}

We undouble a doubled output relation $S'$ into $S$ as follows:
\begin{lstlisting}[style=Prolog-pygsty,label=lst:decoding,escapeinside={[]}]
  T(-x, !) [] :- S['](-x).
  T(-x, @y/-z) :- T(-x/@y/@y, -z).
  S(-x) :- T(!, -x).
  \end{lstlisting}
\end{proof}

\subsection{Intermediate predicates}

The following result is straightforward: intermediate predicates
can be eliminated by folding in the bodies of the intermediate
rules, using equations to unify calling predicates with
intermediate head predicates.

\begin{theorem}\label{thm:intermediate}
    $\I$ is redundant in the presence of $\E$ and the absence of $\N$
    and $\R$.
\end{theorem}

\section{Inexpressibility results} \label{secprim}

In this section we show various inexpressibility results that lead
to absolute or relative primitive results for various features.

\subsection{Recursion}

To see that recursion is primitive also in the context of
Sequence Datalog, we can make the following observation.

\begin{lemma}\label{lem:upperbound}
  Let $Q$ be a query that can be computed by a nonrecursive program.
  Then for any input instance $\inst$, the lengths of paths in $Q(\inst)$
  are bounded by a linear function of the maximal length of a
  path in $\inst$.
\end{lemma}
\begin{proof}
Let $\pr P$ be a nonrecursive program computing a query $Q$.
Let $\pr{P'}$ be $\pr P$ with all negated literals removed.  The
$Q'$ query computed by $\pr{P'}$ contains $Q$, so if we can prove
the claim for $Q'$, it also holds for $Q$.

By Theorem \ref{thm:intermediate}, we know that $Q'$ is computable by a
  program $\pr{P''}$ that does not use intermediate predicates.
  Let $n$ be the number of rules, and for
  $i= 1, \ldots, n$, let $S(e_i)$ be the head of the $i$th rule;
$a_i$ the number of path variables in $e_i$; and $b_i$
  the number of atomic values and variables in $e_i$.
  Then the length of sequences returned by the $i$th rule is at
  most $a_i x + b_i$, with $x$ the maximal length of a sequence in the input.
  The desired linear function can now be taken to be $ax+b$, where
  $a={\rm max}\{a_i \mid 1 \leq i \leq n\}$ and
  $b={\rm max}\{b_i \mid 1 \leq i \leq n\}$.
\end{proof}

We immediately get:

\begin{proposition}\label{prop:recursion2}
Let $a$ be a fixed atomic value and let $Q$ be any query from
$\{R\}$ to $S$ satisfying the property that for every instance
$\inst$ and every natural number $n$ such that $R(a^n) \in
\inst$, the string $a^{n^2}$ is a substring of a path in $Q(I)$.
Then $Q$ is not expressible without recursion.
\end{proposition}

We readily obtain:

\begin{theorem}\label{thm:recursion}
    Recursion is primitive.
\end{theorem}
\begin{proof}
First, we show that $\R$ is primitive in the presence of $\I$.
Consider the following recursive program $\pr P$, computing the
query $Q$ returning all
paths $a^{n^2}$ where $n$ is a natural number such that $R(a^n)$
is in the input:
  \begin{lstlisting}[style=Prolog-pygsty,label=lst:primR1]
  T(!, -x, -x) :- R(-x).
  T(-y/-x, -x, -z) :- T(-y, -x, a/-z).
  S(-y) :- T(-y, -x, !).
  \end{lstlisting}
By Proposition~\ref{prop:recursion2}, query $Q$ is not
expressible without recursion.

The above program uses intermediate predicates.  In the absence
of this feature, consider just the program $\pr{P'}$ consisting
of the first two rules.  Strictly, this program does
not compute a query, as $T$ is ternary.  However, we can turn
$\pr{P'}$ into a program $\pr{P''}$ using the arity simulation
technique of Lemma~\ref{lemenc}.  Program $\pr{P''}$
computes a well-defined query $Q''$ from $\{R\}$ to $T$.
Although $Q''$ is not a natural query,
Proposition~\ref{prop:recursion2} applies to it, so it is not
expressible without recursion.
\end{proof}

\subsubsection{Boolean queries}\label{sub:boolQ}

The above queries showing primitivity of recursion are unary.
What about boolean queries?  It turns out that for boolean
queries, in the presence of intermediate predicates,
recursion is still primitive.  In the absence of intermediate
predicates, however, recursion is redundant for boolean queries,
for trivial reasons.

Let us go in a bit more detail.  Let $R$ be a binary relation
viewed a directed graph.  Let $Q_{\rm a \rightarrow b?}$ be the
boolean query from $\{R\}$ to $S$ that checks whether $b$ is
reachable from $a$.  It is well-known that $Q_{\rm a \rightarrow
b?}$, as a classical relational query, is not computable in
classical Datalog without recursion.
We can view $Q_{\rm a\to b?}$ as a query on sequence databases by encoding
edges $(a,b)$ by paths $a \conc b$ of length two. Under this
encoding, the query is clearly computable by a Sequence Datalog
program in the fragment $\{\I,\R\}$:
  \begin{lstlisting}[style=Prolog-pygsty,label=lst:primR3]
  T(@x/@y) :- R(@x/@y).
  T(@x/@z) :- T(@x/@y),R(@y/@z).
  S :- T(a/b).
  \end{lstlisting}
We can now show that $Q_{\rm a \rightarrow b?}$ is not computable
without recursion in Sequence Datalog by showing that, on input
instances containing only sequences of length two, any
nonrecursive Sequence Datalog program can be simulated by a
classical nonrecursive Datalog program.  This simulation is
similar to the one shown in Lemma~\ref{lem:seq_to_c} appearing
later.  The only added complication is that, due to intermediate
predicates, sequences of lengths longer than two can appear.
However, since there is no recursion, these lengths are bounded
by a constant depending only on the program.

In the absence of the $\I$-feature, we note that any boolean
query, computed by a recursive program without intermediate
predicates, is already computed by the nonrecursive rules only.
Indeed, if the result of the query is \Fl{}, then none of the
rules is fired.  If, on the other hand, the result of the query
is \Tr{}, then at least one rule is fired; however, no recursive
rule can be fired before at least one nonrecursive rule is fired.

\subsection{Intermediate predicates}

It is well known that in classical Datalog, without intermediate
predicates, we can not express queries that require universal
quantifiers \cite{ch_82}.  We can transfer this result to
Sequence Datalog by a simulation technique.

Let $\sch$ be a monadic schema and let $\inst$ be an instance of $\sch$.
We say that $\inst$ is ``two-bounded'' if only paths of
lengths one or two occur in $I$.
We can encode two-bounded instances by classical instances as follows.
Let $\sch^c$ (`c' for classical) be the schema that has two relation names
$R^1$ and $R^2$ for each $R \in \sch$.
For $\inst$ two-bounded as above, we define the classical instance $\inst^c$ of
$\sch^c$ as follows:
\begin{itemize}
    \item $\inst^c(R^1)=\{a \in \dom \mid a \in \inst(R)\}$;
    \item $\inst^c(R^2)=\{(a,b) \mid a\conc b \in \inst(R)\}$.
\end{itemize}

\begin{lemma}\label{lem:seq_to_c}
  Let $\pr P$ be a program in the fragment
  $\{\E, \N, \R\}$, with IDB
  relation name $S$, such that the result of\/ $\pr P$ on a two-bounded instance
  is still two-bounded.
  Then there exists a semipositive classical Datalog program $\pr{P^c}$ using
  only the IDB relation names $S^1$ and $S^2$, such that for every two-bounded
  instance $\inst$ of\/ $\sch$, we have
  ${\pr{P^c}}(\inst^c) = (\pr{P}(\inst))^{\bf c}$.
\end{lemma}
\begin{proof}
Our goal is to eliminate path variables as well as concatenations
in path expressions.  We start with path variables.  In any rule
containing a head predicate or positive predicate of the form
$S(e_1 \conc \pvar x \conc e_2)$ or $R(e_1 \conc \pvar x \conc
e_2)$, we can replace $\$x$ either by $\emp$, $@x$, or $@x_1\cdot
@x_2$ (splitting the rule in three versions).

Path variables may still occur in equations.  By safety, they
must appear in positive equations, and inductively we may assume
that any remaining path variable $\$x$ occurs in a positive
equation $e_1=e_2$ where $e_1$ contains no path variables.
This equation is then of the form
$a_1 \cdots a_n = b_1 \cdots b_m \conc
    \pvar x \conc e$, where the $a$s and $b$s are atomic variables or values.
    \begin{itemize}
      \item If $m = n$, replace $\pvar x$ by the empty path.
      \item If $m > n$, the equation is unsatisfiable and the rule can be
      removed.
      \item If $m < n$, replace $\pvar x$ by
$a_{m+1} \cdots a_i$, for $m < i \leq n$ (splitting the rule in
$n-m+1$ versions).
    \end{itemize}

After these steps, all equations (positive or negated) are of the
form $a_1 \cdots a_n = b_1 \cdots b_m$, where the $a$s and $b$s
are atomic variables or values.  Such equations can be easily
eliminated.  Moreover, any predicates, possibly
negated, that are of the form $R(e)$ with $e$ empty or strictly longer
than two, can be eliminated as well.

We finally replace every remaining predicate (head or body) of the form
$R(a)$ by $R^1(a)$ and every predicate of the form $R(a_1 \conc
a_2)$ by $R^2(a_1, a_2)$, and we are done.
\end{proof}

As a consequence, the query computed by the
following program, belonging to the
fragment $\{\I,\N\}$, cannot be expressed without intermediate
predicates:
  \begin{lstlisting}[style=Prolog-pygsty,label=lst:doubleNeg1]
  W(@x) :- R(@x/@y), + B(@y).
  S(@x) :- R(@x/@y), + W(@x).
  \end{lstlisting}
Indeed, the classical counterpart of this query is the query
asking, on any directed graph where some nodes are ``black'',
for all nodes with only edges to black nodes.
That query is well-known not to be expressible in semipositive Datalog
\cite{ch_82}.

We thus obtain:

\begin{theorem}\label{thm:intermediateNPrim}
    $\I$ is primitive in the presence of $\N$.
\end{theorem}

We also have the following primitivity result in the presence of
recursion.  The proof merely combines some
observations we have already made.

\begin{theorem}\label{thm:intermediateRPrim}
    $\I$ is primitive in the presence of $\R$.
\end{theorem}
\begin{proof}
Recall the squaring query $Q$ from the proof of
Theorem~\ref{thm:recursion}, which is expressible in the fragment
$\{\I,\R\}$.  Suppose, for the sake of contradiction, that $Q$
can be computed by a program without intermediate predicates.
Consider the behavior of this program on the family of singleton
instances $I_n=\{R(a^n)\}$, for all natural numbers $n$.
Since $Q(I_n)$ is nonempty, at least one of the rules must fire,
which is only possible if at least one of the nonrecursive rules
fires.
Since there are no intermediate predicates, however,
this nonrecursive rule must already
produce the correct output $S(a^{n^2})$.
This contradicts Lemma~\ref{lem:upperbound}.
\end{proof}

\subsection{Equations}

The two theorems in the previous subsection provide counterparts
to Theorem~\ref{thm:intermediate}.  The following theorem confirms
that the presence of equations is necessary for
Theorem~\ref{thm:intermediate}, and implies that the fragments
$\{\I\}$ and $\{\E\}$ are actually equivalent.

\begin{theorem}\label{thm:equationsPrim}
    $\E$ is primitive in the absence of $\I$.
\end{theorem}

This result follows immediately from the following lemma.

\begin{lemma}\label{prop:onlyA}
Let $a$ be an atomic value.  The boolean query that checks if
the input relation $R$ contains a path consisting exclusively
of $a$'s, cannot be computed by a program that lacks features
$\I$ and $\E$.
\end{lemma}
\begin{proof}
By the redundancy of packing and arity, we may ignore these
features.  Also, in Section~\ref{sub:boolQ}, we already noted
that in the absence of intermediate predicates,
recursion does not help in expressing boolean queries. Hence,
it suffices to show that
the query cannot be computed by a program in the fragment
$\{\N\}$.  For the sake of contradiction, suppose such a program exists.

Take any rule from the program, and consider the instance $J$
obtained from the positive predicates in the body by ``freezing''
all variables, i.e., viewing them as atomic values distinct from
the atomic values already occurring in the rule.  Unless the rule
is unsatisfiable (in which case we may ignore it), it will fire
on $J$.  So the query is true on $J$ and
the body must contain a positive predicate of the
form $R(a^\ell)$.

Now consider the instance $I=\{R(a^n)\}$ where $n$ is strictly
larger than all values $\ell$ as above found in the rules.  Then
no rule can fire on $I$, but the query is true on $I$, so we have
the desired contradiction.
\end{proof}

Indeed, that query is readily expressed using an equation, as we
well know.

\section{Putting it all together} \label{sechasse}

The results from the previous two sections allow us to
characterize the subsumption relation among fragments (defined in
Section~\ref{secquery}) as follows.

\begin{theorem} \label{thm:subs}
For any fragments $F_1$ and $F_2$, we have $F_1 \subs F_2$
if and only if the following five conditions are satisfied:
\begin{enumerate}
\item $\N \in F_1 \Rightarrow \N \in F_2$;
\item $\R \in F_1 \Rightarrow \R \in F_2$;
\item $\E \in F_1 \Rightarrow (\E \in F_2 \lor \I \in F_2)$;
\item $(\I \in F_1 \land \R \not \in F_1 \land \N \not \in F_1)
\Rightarrow (\I \in F_2 \lor \E \in F_2)$;
\item $(\I \in F_1 \land (\R \in F_1 \lor \N \in F_1))
\Rightarrow \I \in F_2$.
\end{enumerate}
\end{theorem}
\begin{figure*}
  \centering
  \begin{tikzpicture}[
  grow=right,
  every node/.style={font=\ttfamily\small,align=center},
  level 1/.style={level distance=30mm, sibling distance=50mm},
  level 2/.style={level distance=35mm, sibling distance=20mm},
  level 3/.style={level distance=40mm, sibling distance=15mm},
  level 4/.style={level distance=40mm, sibling distance=10mm},
  sloped=true,
  punkt/.style={rectangle, draw=blue!40!black!60, thick}
  ]
  \node {{$\N \isIn F_1$}\\or\\{$\R \isIn F_1$}}
    child {
      node {{$\I \isIn F_1$}\\or\\{$\E \isIn F_1$}}
      child {
        node [punkt] {$\hat{F_1} \subs \hat{F_2}$}
        \noB
      }
      child {
        node[punkt] [rectangle split, rectangle split, rectangle split parts=2, text ragged] {
        {by 3 and 4}
        \nodepart{second} {$\I \in F_2$ or $\E \in F_2$}
        }
        child {
          node[punkt] [rectangle split, rectangle split, rectangle split parts=2, text ragged] {
          {by thm \ref{thm:equatons} and \ref{thm:intermediate}}
          \nodepart{second} {$\hat{F_1} \subs \{\E\} \subs \hat{F_2}$}
          }
          edge from parent node [below] {$\E \in F_2$}
        }
        child {
          node[punkt] [rectangle split, rectangle split, rectangle split parts=2, text ragged] {
          {by thm \ref{thm:equatons}}
          \nodepart{second}  {$\hat{F_1} \subs \{\I\} \subs \hat{F_2}$}
          }
          edge from parent node [above] {$\I \in F_2$}
        }
        \yesB
      }
      \noB
    }
    child {
      node {$\I \isIn F_2$}
      child {
        node {$\E \isIn F_1$}
        child {
          node[punkt] [rectangle split, rectangle split, rectangle split parts=2, text ragged] {
          {by 1 and 2}
          \nodepart{second} {$\hat{F_1} \subs \hat{F_2}$}
          }
          \noB
        }
        child {
          node[punkt] [rectangle split, rectangle split, rectangle split parts=2, text ragged] {
          {by 1 and 2}
          \nodepart{second} {$\hat{F_1} \subs \hat{F_2}$}
          }
          \yesB
          node [above] {by 3, $\E \in F_2$}
        }
        \noB
        node [below] {by 5, $\I \not \in F_1$}
      }
      child {
        node[punkt] [rectangle split, rectangle split, rectangle split parts=2, text ragged] {
        {by 1, 2, and thm \ref{thm:equatons}}
        \nodepart{second} {$\hat{F_1} \subs (\hat{F_1} \cup \{\I\} - \{\E\})
        \subs \hat{F_2}$}
        }
        \yesB
      }
      \yesB
    }
  ;
  \end{tikzpicture}
  \caption{If-direction of Theorem~\ref{thm:subs}.}
  \label{proofsubs}
\end{figure*}

\begin{proof}
For the only-if direction, we verify the five conditions,
assuming $F_1 \subs F_2$.
  \begin{enumerate}
  \item Immediate from the primitivity of negation. We have not
  stated this primitivity as a theorem because it is so clear
  (any fragment without negation can express only monotone
  queries; with negation we can express set difference which is
  not monotone).
  \item Immediate from primitivity of recursion.
  \item Immediate from Theorem~\ref{thm:equationsPrim}.
  \item Assume $\I \in F_1 \land \R \not \in F_1 \land \N \not \in F_1 \land \E
  \not \in F_2 \land \I \not \in F_2$.  By Theorem~\ref{thm:equatons}, we have
  $\{\E\} \leq F_1$.  Now Theorem~\ref{thm:equationsPrim} leads to a
  contradiction with $F_1 \leq F_2$.
  \item Immediate from Theorems \ref{thm:intermediateNPrim} and
  \ref{thm:intermediateRPrim}.
  \end{enumerate}

For the if-direction, since arity
and packing are redundant, $F_1 \leq F_2$ if and only if
  $\Hat{F_1} \subs \Hat{F_2}$, where $\hat{F} = F - \{\A,\P\}$.
Now Figure~\ref{proofsubs} (on page~\pageref{proofsubs}) infers
  $\Hat{F_1} \subs \Hat{F_2}$ from the five conditions and the
  redundancy results.
\end{proof}

\section{Sequence relational algebra} \label{secalg}

Given the importance of algebraic query plans for database query
execution, we show here how to extend the classical relational
algebra to obtain a language equivalent to recursive-free
Sequence Datalog programs.  We note that a similar
language, while calculus-based rather than algebra-based, is
the language StriQuel proposed by Grahne and Waller
\cite{gw_sql}.

The relational algebra, with operators projection; equality
selection; union; difference; and cartesian
product, is well known \cite{ahv_book,ullman}.  To extend this
algebra to our data model (Section~\ref{secmodel}),
we generalize the selection and projection operators and
add two extraction operators.  Let $R$ be an $n$-ary relation.
\begin{description}
\item[Selection:] The classical
equality selection $\sigma_{\$i=\$j}(R)$, with
$i,j\in\{1,\dots,n\}$, returns $\{t \in R \mid t_i = t_j\}$.  We
now allow path expressions $\alpha$ and $\beta$
over the variables $\$1$, \dots,$\$n$ and have the selection
operator $$ \sigma_{\alpha=\beta}(R) := \{t \in R
\mid t(\alpha)=t(\beta)\}. $$ Here, $t$ is viewed as the valuation that
maps $\$i$ to $t_i$ for $i=1,\dots,n$.
\item[Projection:] For path expressions $\alpha_1$, \dots,
$\alpha_p$ over variables $\$1,\dots,\allowbreak \$n$ as above, we define
$$ \pi_{\alpha_1,\dots,\alpha_p}(R) :=
\{(t(\alpha_1),\dots,t(\alpha_p)) \mid t \in R\}. $$
\item[Unpacking:] For $i \in \{1,\dots,n\}$, the operator
${\unpacko}_i(R)$ returns $$\{(t_1,\dots,t_{i-1},s,t_{i+1},\dots,t_n) \mid
(t_1,\dots,t_{i-1},\pack{s},t_{i+1},\dots,t_n) \in R\}. $$
\item[Substrings:] $ \subu_i(R)$ equals $$ \{(t_1,\dots,t_n,s) \mid
(t_1,\dots,t_n) \in R \text{ and $s$ is a substring of $t_i$}\}.
$$
\end{description}

We could also have defined a more powerful unpacking operator,
which extracts components from paths using path expressions,
similar to the use of path expressions in Sequence Datalog.
Such an operator is useful in practice
but can for theoretical purposes be simulated using the given
operators, as we will show.

``Sequence relational algebra'' expressions over a schema $\sch$,
built up using the above operators from the relation names of
$\sch$ and constant relations, are defined as usual.  We have, as
expected, the following theorem.  Note that this result applies
for arbitrary instances, not only for flat inputs and flat
outputs.

\begin{theorem}
For every program $\pr{P}$ without recursion and every IDB
relation name $T$, there exists a sequence relational
algebra expression $E$ such
that for every instance $I$, we have $\pr P(I)(T)=E(I)$.  The
converse statement holds as well.
\end{theorem}

That sequence relational algebra can be translated to Sequence
Datalog is clear.  Our approach to
translate in the other direction is for the most part standard.
We can make use of the following normal form.

\begin{lemma}\label{lemmaforms}
    Let $\pr P$ be a nonrecursive Sequence Datalog program that does not use
    equations.  Then there is a nonrecursive program $\pr P'$ computing the same
    query as $\pr P$ where each rule in $\pr P'$ has one of the following six
    forms:

		\begin{enumerate}
			\item $R_1(v_1, \ldots, v_n) \leftarrow R_2(e_1, \ldots, e_m)$;

			\item $R_1(v_1, \ldots, v_n, e) \leftarrow R_2(v_1, \ldots, v_n)$;

			\item $R_1(v_1,\ldots,v_n)\leftarrow R_2(x_1,\ldots,x_k),
      R_3(y_1,\ldots,y_\ell)$;

			\item $R_1(v_1,\ldots,v_n)\leftarrow R_2(v_1,\ldots,v_n),\neg
      R_3(v'_1,\ldots,v'_{m})$;

			\item $R_1(v'_1,\ldots,v'_m)\leftarrow R_2(v_1,\ldots,v_n)$;

			\item $R(p) \gets {}$.
		\end{enumerate}
The following restrictions apply:
    \begin{itemize}
      \item In all forms, $v_1,\ldots,v_n$ are distinct variables.
      Moreover, in forms 2 to 6, each $v_i$ must be a path variable.
      \item In form 3, the $x_i$ and $y_j$ are path variables and
      $\{v_1,\ldots,v_n\}$ is contained in $\{x_1, \ldots, x_k\} \cup
      \{y_1, \ldots, y_\ell\}$.
      \item In forms 4 and 5, $v'_1,\ldots,v'_m$ are distinct variables taken
      from $\{v_1,\ldots,v_n\}$.
      \item In form 6, $p$ is a path (constant relation).
    \end{itemize}
\end{lemma}

This lemma is stated for programs without equations, since we
know that equations are redundant in the presence of intermediate
predicates.
Given the normal form, extraction rules of the first form can be
expressed in the algebra as follows. First, by compositions of unpacking and
substring operations, we can generate all subpaths until the
maximum packing depth of the expressions appearing in the rule.
Using cartesian product and selection, we then select the
desired paths.  Rules of the second form are generalized
projections.  The other rules are handled as in the classical
relational algebra.  It remains to prove:

\begin{proof}[Proof of Lemma~\ref{lemmaforms}]
The conversion to normal form is best described on a general
example.  Consider the following one-rule Sequence Datalog program:
\begin{lstlisting}[style=Prolog-pygsty]
 T(a/b/c, @x/c/-y, -z/-z):- P1(-y/-y, -z/a, @u/d), P2(-z/@x/c, d), 
                            + N1(@x/-y/-z, a/@x), + N2(a/b, -y).
\end{lstlisting}
  In what follows, we call the rule that we process the main rule and its
  stratum the main stratum.

\paragraph{Step 1: Get variables from positive literals}

\paragraph{Step 1.1}
Replace every occurrence of a positive
      atom $P(e_1, \ldots, e_m)$ by a new predicate $H(v_1, \ldots, v_n)$ where
      $\{v_1, \ldots, v_n\}$ is the set of variables used in the
      atom.
For each $H$ add a new rule of the form $H(v_1, \ldots, v_n)
      \leftarrow P(e_1, \ldots, e_m)$. Note that these set of rules are
      guaranteed to be form~1. Moreover, every \emph{atomic} variable in the
      main rule should be replaced by a new \emph{path} variable.

In case the positive atom does not use variables, then we replace every
occurrence by a new predicate $H(\pv v)$ with a fresh variable $\pv v$.
To get this $H$, we add the two rules $H' \leftarrow P(e_1, \ldots, e_m)$ and
$H(a) \leftarrow H'$ for a new predicate $H'$ and $a \in \dom$.  Note that the
first rule is of form~1, while the second added rule is of form~2.
\begin{lstlisting}[style=Prolog-pygsty]
 H1(-y, -z, @u):- P1(-y/-y, -z/a, @u/d).
 H2(-z, @x):- P2(-z/@x/c, d).
 T(a/b/c, -x/c/-y, -z/-z):- H1(-y, -z, -u), H2(-z, -x), 
                            + N1(-x/-y/-z, a/-x), + N2(a/b, -y).
\end{lstlisting}

\paragraph{Step 1.2}
			\begin{itemize}
				\item When no positive atoms exist in the main rule, then the rule has
        no variables. Only in this case, we add to the main stratum a new rule
        of the form $R(a) \gets {}$, where $R$ is a new
        relation name and $a$ is some value from the domain. This added rule is
        of form~6.  Moreover, we add $R(\pv v)$ to the body of the main rule,
        where $\pv v$ is a fresh path variable.

				\item Otherwise, this step should be repeated until only one positive
        atom remains in the main rule. We remove two positive atoms $H_i(x_1,
        \ldots, x_n)$ and $H_j(y_1, \ldots, y_m)$, and replace them with
        $H(v_1, \ldots, v_k)$, where $H$ is a fresh predicate name, and the set
        of variables $v$s is the union of the set of $x$s and $y$s. In
        addition, we introduce a new rule of the form $$ H(v_1, \ldots, v_k) \gets {}  H_i(x_1, \ldots, x_n), H_j(y_1, \ldots, y_m) $$ in the main
        stratum. This rule is of form~3.
			\end{itemize}
\begin{lstlisting}[style=Prolog-pygsty]
 H1(-y, -z, @u):- P1(-y/-y, -z/a, @u/d).
 H2(-z, @x):- P2(-z/@x/c, d).
 H(-y, -z, -u, -x):- H1(-y, -z, -u), H2(-z, -x).
 T(a/b/c, -x/c/-y, -z/-z):- H(-y, -z, -u, -x), 
                            + N1(-x/-y/-z, a/-x), + N2(a/b, -y).
\end{lstlisting}

\paragraph{Step 2: Separate each negative literal in an
intermediate rule.}

\paragraph{Step 2.1} Let $H(v_1, \ldots, v_n)$ be the only
      positive atom in the body of the rule. Every
      literal
$\lnot N(e_1, \ldots, e_m)$ is replaced by a predicate $HN(v_1, \ldots,
      v_n)$, where $HN$ is a new relation name. Moreover, we add a rule of the
      form $$ HN(v_1, \ldots, v_n) \gets {} H(v_1, \ldots, v_n),
      \lnot N(e_1, \ldots, e_m) $$ to the main stratum, and we remove
      $H(v_1, \ldots, v_n)$ from the main rule.
\begin{lstlisting}[style=Prolog-pygsty]
 H1(-y, -z, @u):- P1(-y/-y, -z/a, @u/d).
 H2(-z, @x):- P2(-z/@x/c, d).
 H(-y, -z, -u, -x):- H1(-y, -z, -u), H2(-z, -x).
 HN1(-y, -z, -u, -x):- H(-y, -z, -u, -x), + N1(-x/-y/-z, a/-x).
 HN2(-y, -z, -u, -x):- H(-y, -z, -u, -x), + N2(a/b, -y).
 T(a/b/c, -x/c/-y, -z/-z):- HN1(-y, -z, -u, -x), 
                             HN2(-y, -z, -u, -x).
\end{lstlisting}

\paragraph{Step 2.2} We do the same as in step 1.2, leaving us in the end with a
      single positive atom holding the variables from the original rule. All
      the rules introduced by this step are of form~3.
\begin{lstlisting}[style=Prolog-pygsty]
 H1(-y, -z, @u):- P1(-y/-y, -z/a, @u/d).
 H2(-z, @x):- P2(-z/@x/c, d).
 H(-y, -z, -u, -x):- H1(-y, -z, -u), H2(-z, -x).
 HN1(-y, -z, -u, -x):- H(-y, -z, -u, -x), + N1(-x/-y/-z, a/-x).
 HN2(-y, -z, -u, -x):- H(-y, -z, -u, -x), + N2(a/b, -y).
 HN(-y, -z, -u, -x):- HN1(-y, -z, -u, -x), HN2(-y, -z, -u, -x).
 T(a/b/c, -x/c/-y, -z/-z):- HN(-y, -z, -u, -x).
\end{lstlisting}

\paragraph{Step 3: Generate negated expressions.}

		We next work on the rules that were introduced to deal with the negated
    atoms.

\paragraph{Step 3.1}
In step~2.1 we added rules with negative literals:
\[HN(v_1, \ldots, v_n) \gets {} H(v_1, \ldots, v_n), \lnot N(e_1, \ldots, e_m)\]
For each such added rule, we define a sequence of rules in order
      to generate the values for the expressions $e_i$. Since our rule is safe
      from the beginning, we are guaranteed that all the variables used in these
      expressions are among the $v$s.

			We inductively generate $m$ rules as follows (where the $v'$s are fresh
      variables) and add them to the main stratum:
			\begin{enumerate}
				\item $N_1(v_1, \ldots, v_n, e_1) \gets {}  H(v_1, \ldots, v_n)$
				\item for $1 < i \leq m$, the rule
			\end{enumerate}
\[N_i(v_1, \ldots, v_n, v'_1, \ldots, v'_{i-1}, e_i) \gets {}   N_{i-1}(v_1, \ldots, v_n, v'_1, \ldots, v'_{i- 1}).\]

Each one of the above rules is of form~2. In addition, we replace
      $H(v_1, \ldots, v_n)$ in the rule under consideration by
$$ N_m(v_1, \ldots, v_n, v'_1, \ldots, v'_{m}). $$
Moreover, we replace $\lnot N(e_1,
      \ldots, e_m)$ by $\lnot N(v'_1, \ldots, v'_m)$.
\begin{lstlisting}[style=Prolog-pygsty]
 H1(-y, -z, @u):- P1(-y/-y, -z/a, @u/d).
 H2(-z, @x):- P2(-z/@x/c, d).
 H(-y, -z, -u, -x):- H1(-y, -z, -u), H2(-z, -x).
 N11(-y, -z, -u, -x, -x/-y/-z):- H(-y, -z, -u, -x).
 N21(-y, -z, -u, -x, a/b):- H(-y, -z, -u, -x).
 N12(-y, -z, -u, -x, -n11, a/-x):- N11(-y, -z, -u, -x, -n11).
 N22(-y, -z, -u, -x, -n21, -y):- N21(-y, -z, -u, -x, -n21).
 HN1(-y, -z, -u, -x):- N12(-y, -z, -u, -x, -n11, -n12), 
                         + N1(-n11, -n12).
 HN2(-y, -z, -u, -x):- N22(-y, -z, -u, -x, -n21, -n22), 
                         + N2(-n21, -n22).
 HN(-y, -z, -u, -x):- HN1(-y, -z, -u, -x), HN2(-y, -z, -u, -x).
 T(a/b/c, -x/c/-y, -z/-z):- HN(-y, -z, -u, -x).
\end{lstlisting}

\paragraph{Step 3.2} We have now obtained rules of the form \[HN(v_1, \ldots,
      v_n) \gets {} N_m(v_1, \ldots, v_n, v'_1, \ldots, v'_m), \lnot N(v'_1,
      \ldots, v'_m).\]
We now further replace them with \[HN(v_1, \ldots, v_n) \gets {}    FN(v_1, \ldots, v_n, v'_1, \ldots, v'_m);\]
where $FN$ is a new
      relation name. Now this rule is of form~5. Moreover, we add the rule
\begin{multline*}
FN(v_1, \ldots, v_n, v'_1, \ldots, v'_m) \gets {} N_m(v_1, \ldots, v_n, v'_1, \ldots, v'_m), \lnot N(v'_1, \ldots, v'_m),
\end{multline*}
      which is of form~4, to the main stratum.
\begin{lstlisting}[style=Prolog-pygsty]
 H1(-y, -z, @u):- P1(-y/-y, -z/a, @u/d).
 H2(-z, @x):- P2(-z/@x/c, d).
 H(-y, -z, -u, -x):- H1(-y, -z, -u), H2(-z, -x).
 N11(-y, -z, -u, -x, -x/-y/-z):- H(-y, -z, -u, -x).
 N21(-y, -z, -u, -x, a/b):- H(-y, -z, -u, -x).
 N12(-y, -z, -u, -x, -n11, a/-x):- N11(-y, -z, -u, -x, -n11).
 N22(-y, -z, -u, -x, -n21, -y):- N21(-y, -z, -u, -x, -n21).
 FN1(-y, -z, -u, -x, -n11, -n12):-
     N12(-y, -z, -u, -x, -n11, -n12), + N1(-n11, -n12).
 FN2(-y, -z, -u, -x, -n21, -n22):- 
     N22(-y, -z, -u, -x, -n21, -n22), + N2(-n21, -n22).
 HN1(-y, -z, -u, -x):- FN1(-y, -z, -u, -x, -n11, -n12).
 HN2(-y, -z, -u, -x):- FN2(-y, -z, -u, -x, -n21, -n22).
 HN(-y, -z, -u, -x):- HN1(-y, -z, -u, -x), HN2(-y, -z, -u, -x).
 T(a/b/c, -x/c/-y, -z/-z):- HN(-y, -z, -u, -x).
\end{lstlisting}

\paragraph{Step 4: Generate final head expressions.}

		We are now left to work on the final rule which is normalized in a similar
    way as step 3.1. The final rule is of the form $T(e_1, \ldots, e_m) \gets {}    H(v_1, \ldots, v_n)$, where by safety it is guaranteed that any
    variable appearing in any of the $e$s is among the $v$s.

		We inductively generate $m$ rules as follows (where the $v'$s are fresh
    variables):
		\begin{enumerate}
			\item $T_1(v_1, \ldots, v_n, e_1) \gets {}  H(v_1, \ldots, v_n)$
			\item for $1 < i \leq m$, the rule
		\end{enumerate}
\[T_i(v_1, \ldots, v_n, v'_1, \ldots, v'_{i-1}, e_i) \gets {}   T_{i-1}(v_1, \ldots, v_n, v'_1, \ldots, v'_{i- 1}).\]

Each one of the above rules is of form~2. The last thing to be done is
to update the main rule to 
\[T(v'_1, \ldots, v'_m) \gets {}    T_m(v_1, \ldots, v_n, v'_1, \ldots, v'_m).\]
Now, this rule is of form~5.
\begin{lstlisting}[style=Prolog-pygsty]
 H1(-y, -z, @u):- P1(-y/-y, -z/a, @u/d).
 H2(-z, @x):- P2(-z/@x/c, d).
 H(-y, -z, -u, -x):- H1(-y, -z, -u), H2(-z, -x).
 N11(-y, -z, -u, -x, -x/-y/-z):- H(-y, -z, -u, -x).
 N21(-y, -z, -u, -x, a/b):- H(-y, -z, -u, -x).
 N12(-y, -z, -u, -x, -n11, a/-x):- N11(-y, -z, -u, -x, -n11).
 N22(-y, -z, -u, -x, -n21, -y):- N21(-y, -z, -u, -x, -n21).
 FN1(-y, -z, -u, -x, -n11, -n12):-
     N12(-y, -z, -u, -x, -n11, -n12), + N1(-n11, -n12).
 FN2(-y, -z, -u, -x, -n21, -n22):- 
     N22(-y, -z, -u, -x, -n21, -n22), + N2(-n21, -n22).
 HN1(-y, -z, -u, -x):- FN1(-y, -z, -u, -x, -n11, -n12).
 HN2(-y, -z, -u, -x):- FN2(-y, -z, -u, -x, -n21, -n22).
 HN(-y, -z, -u, -x):- HN1(-y, -z, -u, -x), HN2(-y, -z, -u, -x).
 T1(-y, -z, -u, -x, a/b/c):- HN(-y, -z, -u, -x).
 T2(-y, -z, -u, -x, -t1, -x/c/-y):- T1(-y, -z, -u, -x, -t1).
 T3(-y, -z, -u, -x, -t1, -t2, -z/-z):- 
     T2(-y, -z, -u, -x, -t1, -t2).
 T(-t1, -t2, -t3):- T3(-y, -z, -u, -x, -t1, -t2, -t3).
\end{lstlisting}
\end{proof}

\section{Conclusion} \label{seconc}

Sequence databases and sequence query processing (e.g.,
\cite{zaniolo_seq}) were an active research topic twenty years
ago or more.  We hope our paper can revive interest in the topic,
given its continued relevance for advanced database applications.
Systems in use today do support sequences one way or another, but
often only nominally, without high expressive power or
performance.  This situation may cause application builders to
bypass the database system and solve their problem in an ad-hoc
manner.

Of course, to support data science, there is much current
research on database systems and query languages for arrays and
tensors, e.g.,
\cite{rusu_survey,tiledb,hutchison,polystorelalg,lara_chile}.
However, in this domain, applications are typically focused on
supporting linear algebra operations
\cite{jermaine_linearalgebra_cacm,hutchison,lara_chile}.  Such
applications are qualitatively different from the more generic
type of sequence database queries considered in this paper.

We note that other sequence query language approaches, not based
on Datalog, deserve attention as well.  There have been proposals
based on functional programming \cite{libkin_arrays}, on
structural recursion \cite{lists_ptime}, and on transducers
\cite{bonnermecca_transducers,blss_strings,grahne_strings,gw_seq}.
On the other hand, a proposal very close in spirit to Sequence
Datalog can be found in the work by Grahne and Waller
\cite{gw_sql} already mentioned in Section~\ref{secalg}.

Sequence Datalog is also a very useful language for dealing with
non-flat instances.  In this paper, for reasons we have
explained, we focused on queries from flat instances to flat
instances.  However, using packing, interesting data structures
can be represented in a direct manner.  For example, a tree with
root label $a$ and childtrees $T_1,\dots,T_n$ can be represented
by the path $a \cdot \langle T_1 \rangle \cdots \langle T_n
\rangle$ (where each $T_i$ is represented by a path
in turn).  Thus, Sequence Datalog can be used as an XML-to-XML
query language and more.

We conclude by recalling an intriguing theoretical open problem
already mentioned before \cite{hidders2017j}.  It can be stated
independently of Sequence Datalog, although we did stumble upon
the problem while thinking about Sequence Datalog.  Consider
monadic Datalog with stratified negation over sets of \emph{natural
numbers,} with natural number constants and variables, and
addition as the only operation.  Which functions on finite sets of
natural numbers are expressible in this language?

\bibliographystyle{plainurl}
\bibliography{database}

\end{document}

%% file: macros.tex
\usepackage{listings}
\usepackage{tikz}
\usetikzlibrary{shapes}
\usetikzlibrary{arrows,patterns}
\usepackage{tikz-qtree}
\usetikzlibrary{matrix}
\usepackage{enumitem}

\newcommand{\pack}[1]{\langle {#1} \rangle}
\newcommand{\sch}{\Gamma}
\newcommand{\dom}{{\mathbf{dom}}}
\newcommand{\inst}{I}

\newcommand{\emp}{\epsilon}
\newcommand{\conc}{\cdot}
\newcommand{\paths}{\Pi}
\newcommand{\avar}{@}
\newcommand{\pvar}{\$}

\newcommand{\feat}{\Phi}
\newcommand{\subs}{\leq}
\newcommand{\nsubs}{\nleq}
\newcommand{\enc}[2]{#1 \conc a \conc #2 \conc a \conc #1 \conc b \conc #2}
\newcommand{\ru}[2]{#1 \leftarrow #2}

\newcommand{\weq}=
\newcommand{\stratum}{\Delta}
\newcommand{\Tr}{true}
\newcommand{\Fl}{false}

\newcommand{\pr}[1]{{\bf #1}}
\newcommand{\ps}[1]{\delta(#1)}

\newcommand{\isIn}{\stackrel{?}{\in}}

\newcommand{\A}{{\textsf{A}}}
\newcommand{\N}{{\textsf{N}}}
\newcommand{\I}{{\textsf{I}}}
\newcommand{\R}{{\textsf{R}}}
\renewcommand{\P}{{\textsf{P}}}
\newcommand{\E}{{\textsf{E}}}

\newcommand{\subu}{\mathit{SUB}}

\newcommand{\unpacko}{\mathit{UNPACK}}

\newcommand{\pv}[1]{\pvar #1}

\newcommand{\yesB}{edge from parent node [below] {yes}}
\newcommand{\noB}{edge from parent node [above] {no}}

\newcommand{\failN}[4]{
\node at (#1,#2) (#3) [align=center]
{#4};}

%% file: packing-example.tex
\begin{figure*}
  \centering
  \resizebox{0.85\textwidth}{!}{%
  \rotatebox{270}{
  \begin{tikzpicture}
\node at (-2,1) (Root) [align=center]
{$\pvar x \conc \pack{\avar y \conc \pvar z} \conc \avar w
$\\=\\$\pvar u \conc \pvar v \conc \pvar u$};

\node at (2,-3) (n4) [align=center]
{$\pack{\avar y \conc \pvar z} \conc \avar w$
\\=\\$\pvar u \conc \pvar v \conc \pvar x \conc \pvar u$};

\node at (2,-7) (n9) [align=center]{$\avar w$
\\=\\$\pvar v \conc \pvar x \conc \pack{\avar y \conc \pvar z}$};

{\failN{0}{-10}{n11}{$\emp$
\\=\\$\pvar x \conc \pack{\avar y \conc \pvar z}$}}

{\failN{2}{-12}{n11n}{$\emp$
\\=\\$\pvar v \conc \pvar x \conc \pack{\avar y \conc \pvar z}$}}

\node at (7,-7) (n10) [align=center]{$\avar w$
\\=\\
$\pvar u \conc \pvar v \conc \pvar x \conc \pack{\avar y \conc \pvar z}
\conc \pvar u$};

{\failN{5}{-10}{n12}{$\emp$
\\=\\$\pvar v \conc \pvar x \conc \pack{\avar y \conc \pvar z} \conc
\avar w$}}

{\failN{7}{-12}{n12n}{$\emp = \pvar u \conc \pvar v \conc \pvar x \conc
\pack{\avar y \conc \pvar z} \conc \avar w \conc \pvar u$}}

\node at (-7,-3) (nn2) [align=center]
{$\pvar x \conc \pack{\avar y \conc \pvar z} \conc \avar w $
\\=\\$\pvar v \conc \pvar u$};

\node at (-2,-3) (nn10) [align=center]
{$\pack{\avar y \conc \pvar z} \conc \avar w$
\\=\\$\pvar v \conc \pvar u$};

\node at (-2,-8) (nn16) [align=center]
{$\avar w = \pvar v \conc \pvar u$};

\failN{-2}{-12}{nn20}{$\emp$=$\pvar v \conc \pvar u$}

\node at (-12,-3) (nn12) [align=center]
{$\pvar x \conc \pack{\avar y \conc \pvar z} \conc \avar w$
\\=\\$\pvar u$};

\failN{-14}{-6}{nn5}{
$\pvar x \conc \pack{\avar y \conc \pvar z} \conc \avar w$
\\=\\$\emp$}

\failN{-12}{-8}{nn6}{
$\pack{\avar y \conc \pvar z} \conc \avar w$
\\=\\$\emp$}

\node at (-9,-6) (nn7) [align=center]
{$\pack{\avar y \conc \pvar z} \conc \avar w$
\\=\\$\pvar u$};

\failN{-9}{-10}{nn8}{
$\avar w$=$\emp$}

\node at (-4,-6) (nn9) [align=center]
{$\avar w$=$\pvar u$};

\node[draw, ultra thick, circle] at (-6,-8) (nn13) [align=center]
{$\emp$=$\emp$};


\failN{-4}{-10}{nn19}{$\emp$=$\pvar u$}


\path[->, ultra thick] (Root) edge node[sloped,above]
{$\pvar x \mapsto \pvar u \conc \pvar x$} (nn2);

\path[->, ultra thick] (Root) edge node[sloped,above]
{$\pvar x \mapsto \pvar u$} (nn10);

\path[->] (Root) edge node[sloped,above]
{$\pvar u \mapsto \pvar x \conc \pvar u$} (n4);

\path[->] (n4) edge node[sloped,above]
{$\pvar u \mapsto \pack{\avar y \conc \pvar z}$} (n9);

\path[->] (n9) edge node[sloped,above]
{$\pvar v \mapsto \avar w$} (n11);

\path[->] (n9) edge node[sloped,above]
{$\pvar v \mapsto \avar w \conc \pvar v$} (n11n);

\path[->] (n10) edge node[sloped,above]
{$\pvar u \mapsto \avar w$} (n12);

\path[->] (n10) edge node[sloped,above]
{$\pvar u \mapsto \avar w \conc \pvar u$} (n12n);

\path[->] (n4) edge node[sloped,above]
{$\pvar u \mapsto \pack{\avar y \conc \pvar z} \conc \avar u$} (n10);

\path[->, ultra thick] (nn2) edge node[sloped,below]
{$\pvar x \mapsto \pvar v \conc \pvar x$} (nn12);

\path[->, ultra thick] (nn2) edge node[sloped,above]
{$\pvar x \mapsto \pvar v$} (nn7);

\path[->, ultra thick] (nn2) edge node[sloped,above]
{$\pvar v \mapsto \pvar x \conc \pvar v$} (nn10);

\path[->] (nn12) edge node[sloped,above]
{$\pvar x \mapsto \pvar u \conc \pvar x$} (nn5);

\path[->] (nn12) edge node[sloped,above]
{$\pvar x \mapsto \pvar u$} (nn6);

\path[->, ultra thick] (nn12) edge node[sloped,above]
{$\pvar u \mapsto \pvar x \conc \pvar u$} (nn7);

\path[->] (nn7) edge node[sloped,above]
{$\pvar u \mapsto \pack{\avar y \conc \pvar z}$} (nn8);

\path[->, ultra thick] (nn7) edge node[sloped,below]
{$\pvar u \mapsto \pack{\avar y \conc \pvar z} \conc \pvar u$} (nn9);

\path[->, ultra thick] (nn9) edge node[sloped,below]
{$\pvar u \mapsto \avar w$} (nn13);

\path[->] (nn9) edge node[sloped,above]
{$\pvar u \mapsto \avar w \conc \pvar u$} (nn19);

\path[->, ultra thick] (nn10) edge node[sloped,below]
{$\pvar v \mapsto \pack{\avar y \conc \pvar z}$} (nn9);

\path[->] (nn10) edge node[sloped,above]
{$\pvar v \mapsto \pack{\avar y \conc \pvar z} \conc \pvar v$} (nn16);

\path[->] (nn16) edge node[sloped,below] {$\pvar v \mapsto \avar w$} (nn19);

\path[->] (nn16) edge node[sloped,above]
{$\pvar v \mapsto \avar w \conc \pvar v$} (nn20);

\end{tikzpicture}
 } }%
  \caption{Associative unification on an equation on path
  expressions.  Bold edges indicate the successful branches.}
  \label{fig:packingex1}
\end{figure*}

%% file: main.bbl
\begin{thebibliography}{10}

\bibitem{seqdatalog-pods}
H.~Aamer, J.~Hidders, J.~Paredaens, and J.~Van~den Bussche.
\newblock Expressiveness within sequence datalog.
\newblock In {\em Proceedings 40th36th ACM Symposium on Principles of
  Databases}, pages 70--81. ACM, 2021.

\bibitem{abdulrab1989solving}
H.~Abdulrab and J.-P. P{\'e}cuchet.
\newblock Solving word equations.
\newblock {\em Journal of Symbolic Computation}, 8(5):499--521, 1989.

\bibitem{ahv_book}
S.~Abiteboul, R.~Hull, and V.~Vianu.
\newblock {\em Foundations of Databases}.
\newblock Addison-Wesley, 1995.

\bibitem{datalog2.019}
M.~Alviano and A.~Pieris, editors.
\newblock {\em Datalog 2.0 2019: Third International Workshop on the Resurgence
  of Datalog in Academia and Industry}, volume 2368 of {\em CEUR Workshop
  Proceedings}, 2019.

\bibitem{oomanifesto}
M.~Atkinson, F.~Bancilhon, D.~DeWitt, K.~Dittrich, D.~Maier, and S.~Zdonik.
\newblock The object-oriented database system manifesto.
\newblock In W.~Kim, J.-M. Nicolas, and S.~Nishio, editors, {\em Proceedings
  1st International Conference on Deductive and Object-Oriented Databases},
  pages 40--57. Elsevier Science Publishers, 1989.

\bibitem{lara_chile}
P.~Barcel\'o, N.~Higeura, J.~P\'erez, and B.~Suercaseaux.
\newblock On the expressiveness of {LARA}: A unified language for linear and
  relational algebra.
\newblock In C.~Lutz and J.C. Jung, editors, {\em Proceedings 23rd
  International Conference on Database Theory}, volume 155 of {\em Leibniz
  International Proceedings in Informatics}, pages 6:1--6:20. Schloss
  Dagstuhl--Leibniz-Zentrum f\"ur Informatik, 2020.

\bibitem{datalog2.0}
P.~Barcel\'o and R.~Pichler, editors.
\newblock {\em Datalog in Academia and Industry: Second International Workshop,
  Datalog 2.0}, volume 7494 of {\em Lecture Notes in Computer Science}.
  Springer, 2012.

\bibitem{blss_strings}
M.~Benedikt, L.~Libkin, T.~Schwentick, and L.~Segoufin.
\newblock Definable relations and first-order query languages over strings.
\newblock {\em Journal of the ACM}, 50(5):694--751, 2003.

\bibitem{bonnermecca_sequences}
A.~Bonner and G.~Mecca.
\newblock Sequences, {D}atalog, and transducers.
\newblock {\em Journal of Computer and System Sciences}, 57:234--259, 1998.

\bibitem{bonnermecca_transducers}
A.J. Bonner and G.~Mecca.
\newblock Querying sequence databases with transducers.
\newblock {\em Acta Informatica}, 36:511--544, 2000.

\bibitem{ch_82}
A.K. Chandra and D.~Harel.
\newblock Structure and complexity of relational queries.
\newblock {\em Journal of Computer and System Sciences}, 25(1):99--128, 1982.

\bibitem{chomicki_tql}
J.~Chomicki.
\newblock Temporal query languages: a survey.
\newblock In D.M. Gabbay and H.J. Ohlbach, editors, {\em Temporal Logic:
  {ICTL'94}}, volume 827 of {\em Lecture Notes in Computer Science}, pages
  506--534. Springer-Verlag, 1994.

\bibitem{3genmanifesto}
The Commitee~for Advanced DBMS~Function.
\newblock Third-generation database system manifesto.
\newblock {\em SIGMOD Record}, 19(3):31--44, 1990.

\bibitem{datalogreloaded}
O.~de~Moor, G.~Gottlob, T.~Furche, and A.~Sellers, editors.
\newblock {\em Datalog Reloaded: First International Workshop, Datalog 2010},
  volume 6702 of {\em Lecture Notes in Computer Science}. Springer, 2011.

\bibitem{duran2018associative}
F.~Dur{\'a}n, S.~Eker, S.~Escobar, N.~Mart{\'\i}-Oliet, J.~Meseguer, and
  C.~Talcott.
\newblock Associative unification and symbolic reasoning modulo associativity
  in {Maude}.
\newblock In V.~Rusu, editor, {\em Proceedings 12th International Workshop on
  Rewriting Logic and Its Applications}, volume 11152 of {\em Lecture Notes in
  Computer Science}, pages 98--114. Springer, 2018.

\bibitem{ef_fmt2}
H.-D. Ebbinghaus and J.~Flum.
\newblock {\em Finite Model Theory}.
\newblock Springer, second edition, 1999.

\bibitem{documentspanners}
R.~Fagin, B.~Kimelfeld, F.~Reiss, and S.~Vansummeren.
\newblock Document spanners: A formal approach to information extraction.
\newblock {\em Journal of the ACM}, 62(2):12:1--12:51, 2015.

\bibitem{gw_seq}
S.~Ginsburg and X.S. Wang.
\newblock Regular sequence operations and their use in database queries.
\newblock {\em Journal of Computer and System Sciences}, 56(1):1--26, 1998.

\bibitem{grahne_strings}
G.~Grahne, M.~Nyk\"anen, and E.~Ukkonen.
\newblock Reasoning about strings in databases.
\newblock {\em Journal of Computer and System Sciences}, 59:116--162, 1999.

\bibitem{gw_sql}
G.~Grahne and E.~Waller.
\newblock How to make {SQL} stand for {S}tring {Q}uery {L}anguage.
\newblock In R.C.H. Connor and A.O. Mendelzon, editors, {\em Research Issues in
  Structured and Semistructured Database Programming}, volume 1949 of {\em
  Lecture Notes in Computer Science}, pages 61--79. Springer, 2000.

\bibitem{grohe_arity}
M.~Grohe.
\newblock Arity hierarchies.
\newblock {\em Annals of Pure and Applied Logic}, 82(2):103--163, 1996.

\bibitem{hidders2017j}
J.~Hidders, J.~Paredaens, and J.~Van~den Bussche.
\newblock J-logic: Logical foundations for json querying.
\newblock In {\em Proceedings 36th ACM Symposium on Principles of Databases},
  pages 137--149. ACM, 2017.

\bibitem{hidders2020j}
J.~Hidders, J.~Paredaens, and J.~Van~den Bussche.
\newblock J-logic: a logic for querying json.
\newblock arXiv:2006.04277, 2020.

\bibitem{hutchison}
D.~Hutchison, B.~Howe, and D.~Suciu.
\newblock {LaraDB}: A minimalist kernel for linear and relational algebra
  computation.
\newblock In F.N. Afrati and J.~Sroka, editors, {\em Proceedings 4th ACM SIGMOD
  Workshop on Algorithms and Systems for MapReduce and Beyond}, pages
  2:1--2:10, 2017.

\bibitem{pmanifesto}
{IEEE Task Force on Process Mining}.
\newblock Process mining manifesto, 2011.
\newblock URL: \url{https://www.tf-pm.org/resources/manifesto}.

\bibitem{jo_dmls}
H.V. Jagadish and F.~Olken.
\newblock Database management for life science research.
\newblock {\em SIGMOD Record}, 33(2):15--20, 2004.

\bibitem{polystorelalg}
H.~Jananthan, Z.~Zhou, et~al.
\newblock Polystore mathematics of relational algebra.
\newblock In J.-Y. Nie, Z.~Obradovic, T.~Suzumura, et~al., editors, {\em
  Proceedings IEEE International Conference on Big Data}, pages 3180--3189.
  IEEE, 2017.

\bibitem{zaniolo_streams_tods}
Y.~Law, H.~Wang, and C.~Zaniolo.
\newblock Relational languages and data models for continuous queries on
  sequcnes and data streams.
\newblock {\em ACM Transactions on Database Systems}, 36(2):8:1--8:32, 2011.

\bibitem{g-core}
{LDBC Graph Query Language Task Force}.
\newblock {G-CORE}: A core for future graph query languages.
\newblock In {\em Proceedings 2018 International Conference on Management of
  Data}, pages 1421--1432. ACM, 2018.

\bibitem{libkin_arrays}
L.~Libkin, R.~Machlin, and L.~Wong.
\newblock A query language for multidimensional arrays: design,
  implementations, and optimization techniques.
\newblock In {\em Proceedings of the 1996 {ACM SIGMOD} International Conference
  on Management of Data}, volume 25:2 of {\em SIGMOD Record}, pages 228--239.
  ACM Press, 1996.

\bibitem{jermaine_linearalgebra_cacm}
S.~Luo, Z.J. Gao, M.N. Gubanov, L.L. Perez, D.~Jankov, and C.M. Jermaine.
\newblock Scalable linear algebra on a relational database system.
\newblock {\em Communications of the ACM}, 63(8):93--101, 2020.

\bibitem{meccabonner_termination}
G.~Mecca and A.J. Bonner.
\newblock Query languages for sequence databases: Termination and complexity.
\newblock {\em IEEE Transactions on Knowledge and Data Engineering},
  13(3):519--525, 2001.

\bibitem{spannerlog}
Y.~Nahshon, L.~Peterfreund, and S.~Vansummeren.
\newblock Incorporating information extraction in the relational database
  model.
\newblock In {\em Proceedings 19th International Conference on Web and
  Databases}, pages 6:1--6:7. ACM, 2019.

\bibitem{tiledb}
S.~Papadopoulos et~al.
\newblock The {TileDB} array data storage manager.
\newblock {\em Proceedings of the VLDB Endowment}, 10(4):349--360, 2016.

\bibitem{peterprog}
L.~Peterfreund et~al.
\newblock Recursive programs for document spanners.
\newblock In P.~Barcelo and M.~Calautti, editors, {\em Proceedings 22nd
  International Conference on Database Theory}, volume 127 of {\em LIPIcs},
  pages 13:1--13:18. Schloss Dagstuhl--Leibniz Center for Informatics, 2019.

\bibitem{chili_jsonschema}
F.~Pezoa, J.L. Reutter, F.~Suarez, M.~Ugarte, and D.~Vrgo\v{c}.
\newblock Foundations of {JSON} {S}chema.
\newblock In {\em Proceedings 25th International Conference on World Wide Web},
  pages 263--273, 2016.

\bibitem{plandowski_fswep}
W.~Plandowski.
\newblock On {PSPACE} generation of a solution set of a word equation and its
  applications.
\newblock {\em Theoretical Computer Science}, 792:20--61, 2019.

\bibitem{plotkin1972building}
G.~Plotkin.
\newblock Building-in equational theories.
\newblock In B.~Meltzer and D.~Michie, editors, {\em Machine Intelligence 7},
  pages 73--90. Edinburgh University Press, 1972.

\bibitem{ramasrql}
R.~Ramakrishnan et~al.
\newblock {SRQL:} sorted relational query language.
\newblock In M.~Rafanelli and M.~Jarke, editors, {\em Proceedings 10th
  International Conference on Scientific and Statistical Database Management},
  pages 84--95. IEEE Computer Society, 1998.

\bibitem{lists_ptime}
E.L. Robertson, L.V. Saxton, D.~Van~Gucht, and S.~Vansummeren.
\newblock Structural recursion as a query language on lists and ordered trees.
\newblock {\em Theory of Computing Systems}, 44:590--619, 2009.

\bibitem{rusu_survey}
F.~Rusu and Y.~Cheng.
\newblock A survey on array storage, query languages, and systems.
\newblock arXiv:1302.0103, 2013.

\bibitem{zaniolo_seq}
R.~Sadri, C.~Zaniolo, A.~Zarkesh, et~al.
\newblock Expressing and optimizing sequence queries in database systems.
\newblock {\em ACM Transactions on Database Systems}, 29(2):282--318, 2004.

\bibitem{ramaseq}
P.~Seshadri, M.~Livny, and R.~Ramakrishnan.
\newblock {SEQ}: A model for sequence databases.
\newblock In P.S. Yu and A.L.P. Chen, editors, {\em Proceedings 11th
  International Conference on Data Engineering}, pages 232--239. IEEE Computer
  Society, 1995.

\bibitem{xlog}
W.~Shen et~al.
\newblock Declarative information extraction using {D}atalog with embedded
  extraction.
\newblock In Ch. Koch et~al., editors, {\em Proceedings 33th International
  Conference on Very Large Data Bases}, pages 1033--1044. ACM, 2007.

\bibitem{ullman}
J.D. Ullman.
\newblock {\em Principles of Database and Knowledge-Base Systems}, volume~I.
\newblock Computer Science Press, 1988.

\end{thebibliography}
